\documentclass[
 reprint,
 amsmath,amssymb,
 pra,superscriptaddress, 
 nofootinbib
]{revtex4-2}
\usepackage[caption=false]{subfig}
\usepackage{amsmath}
\usepackage{amsthm}
\usepackage{amssymb}

\usepackage{tabularx}
\usepackage{graphicx}
\usepackage{dsfont}
\usepackage{multirow}
\usepackage{framed}
\usepackage{enumitem}
\usepackage{hyperref}
\usepackage{tikz}
\usepackage{tikz-cd}
\usepackage{mathtools}
\usepackage[framemethod=tikz]{mdframed}
\usetikzlibrary{calc}

\theoremstyle{plain}
\newtheorem*{theorem*}{Theorem}
\newtheorem{theorem}{Theorem}
\newtheorem{result}{Result}
\newtheorem{lemma}{Lemma}

\newtheorem{definition}{Definition}
\newtheorem{corollary}{Corollary}[theorem]

\newcommand{\defeq}{\vcentcolon=}
\newcommand{\eqdef}{=\vcentcolon}
\newcommand{\poly}{\mathrm{poly}}
\newcommand{\rank}{\mathrm{rank}}
\usepackage{verbatim}

\begin{document}

\preprint{APS/123-QED}

\title{Duality constrains optimal thresholds in quantum error correction}

\author{Lucas~H.~English}
\email{lucas.english@sydney.edu.au}
\affiliation{%
School of Physics, University of Sydney, New South Wales 2006, Australia
}
\author{Haoyuan~Luo}
\affiliation{%
School of Physics, University of Sydney, New South Wales 2006, Australia
}%
\author{YangMing~Wang}
\affiliation{%
School of Physics, University of Sydney, New South Wales 2006, Australia
}%
\author{Basudha~Srivastava}
\affiliation{%
Quantinuum, Terrington House, Cambridge, CB2 1NL, UK
}%
\author{Stephen~D.~Bartlett}
\affiliation{%
School of Physics, University of Sydney, New South Wales 2006, Australia
}%
\author{Dominic~J.~Williamson}
\affiliation{%
School of Physics, University of Sydney, New South Wales 2006, Australia
}%

\date{\today}

\begin{abstract}
Error correction thresholds are often treated as the primary figure of merit for comparing quantum error-correcting code families.
We show that the optimal error correction threshold for many commonly considered codes is constrained to a single universal value at leading order in a replica limit.
Through a statistical mechanical mapping, we demonstrate that duality constrains all zero-rate \emph{em}-symmetric CSS codes to have the same optimal code capacity threshold.
Here, \emph{em} symmetry means that the X- and Z-type parity-check matrices are equivalent up to row and column permutations.
Under this statistical mechanical mapping, \emph{em}-symmetric CSS codes are self-dual under a generalized Kramers-Wannier duality up to a mixing of logical sectors.
For zero-rate code families, this mixing contributes only subextensive corrections, so the thermodynamic bulk free energy is self-dual in the trivial logical sector. 
This self-duality fixes the clean critical point and constrains the disordered phase boundary. We also show that self-duality is preserved under code concatenation, and that optimal decoding of concatenated codes can be reformulated as a renormalization group flow on a hierarchical lattice. 
Our results provide a common framework for analyzing topological, concatenated, and more general quantum low-density parity-check code families, including both their optimal code capacity thresholds and their sub-threshold logical error suppression.
\end{abstract}

\maketitle

\section{Introduction}

Fault-tolerant quantum computation is underpinned by families of quantum error-correcting codes, which provide arbitrary suppression of logical errors as a function of code block size \cite{Shor_96,Dennis_02}.
This arbitrary suppression is only possible for physical error rates below a critical value known as the threshold \cite{Aharonov_97, Knill_98, Dennis_02}.
How far the physical error rate is below threshold determines the scaling of logical error suppression with increasing code block size \cite{PhysRevA.88.062308,smith_mitigating_2024}.
Thresholds therefore play a central role in comparing quantum error-correcting code families, and much of the architecture-level discussion of fault tolerance is naturally framed as a search for codes with higher thresholds~\cite{Acharya2025,Bravyi2024,PhysRevLett.103.090501,PhysRevLett.102.200501}.

The code capacity threshold of a code, which describes error correction in a setting where data qubits are subjected to noise but stabilizer check measurements are error-free, can be mapped to the critical point of a disordered statistical mechanical model \cite{Dennis_02,Chubb_2021}. The phase diagrams of the statistical mechanical models associated with the toric and color codes under both bit-flip and depolarizing code capacity noise have been found to agree~\cite{Katzgraber_2013,PhysRevX.2.021004}. We find that this coincidence extends beyond these two code families. The central empirical puzzle motivating this work is illustrated in Fig.~\ref{fig:phase_diagram}. The phase boundary of the toric/planar code (described by the random-bond Ising model under the statistical mechanical mapping)~\cite{Dennis_02,Wang_03,PhysRevB.65.054425}, together with those of the concatenated Steane code and concatenated
surface-17 code \cite{PhysRevA.90.062320}, nearly coincide, and their intersections with the Nishimori line \cite{Nishimori_1981} lead to similar optimal thresholds.  What explains this coincidence for these very different code families?

In this work, we show that this shared behaviour arises from a common self-duality in the associated statistical mechanical models.
We consider Calderbank-Shor-Steane (CSS) codes~\cite{PhysRevA.54.1098,PhysRevLett.77.793,10.1098/rspa.1996.0136} whose X- and Z-type parity-check matrices can be row and column permuted into one another. We say that these codes possess an \emph{em} symmetry~\cite{PhysRevB.110.085158}. 
We show that the statistical mechanical phase boundaries for \emph{em}-symmetric CSS codes with asymptotically zero rate are strongly constrained by duality~\cite{10.1063/1.1665530}. This explains why topological codes, concatenated codes, and certain structured quantum low-density parity-check (qLDPC) families including the abelian two-block group algebra (A2BGA) codes \cite{PhysRevA.109.022407} can exhibit remarkably similar optimal code capacity thresholds despite having very different microscopic constructions.

\begin{figure}[t]
    \centering
    \includegraphics[width=\linewidth]{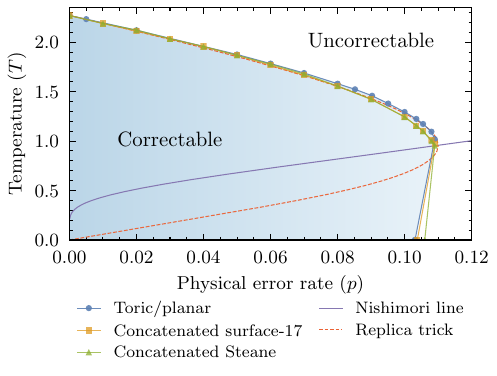}
    \caption{Phase diagram for concatenated and topological QEC codes under bit-flip noise in the code capacity setting. The phase boundary is shown for the toric/planar code family (blue; data taken from Refs.~\cite{PhysRevB.65.054425,Wang_03}), concatenated surface-17 code (orange) and concatenated Steane code (green). The dashed red line indicates the phase boundary prediction obtained via the replica trick. We obtain the phase boundaries for the concatenated codes in this work. The intersection of the Nishimori line with each phase boundary determines the corresponding optimal error threshold. The shaded region denotes the stable ordered phase, in which quantum error correction is feasible and logical failures are arbitrarily suppressed with code size.}
    \label{fig:phase_diagram}
\end{figure}

\subsection{Summary of main results}

Our main results can be summarized as follows. First, we show that a common symmetry of the CSS parity-check matrices strongly constrains the corresponding code capacity decoding problem. Specifically, we show that \emph{em}-symmetric CSS codes are self-dual under a generalized Kramers-Wannier duality~\cite{10.1063/1.1665530,PhysRevA.97.062320} up to mixing of logical sectors. For zero-rate families, this sector mixing is subextensive and vanishes in the thermodynamic limit, so the corresponding clean statistical mechanical models, where all coupling constants are homogeneous, are thermodynamically self-dual. In the decoding problem, this clean model corresponds to the fully postselected limit, in which one aborts and reinitializes the code whenever any nontrivial syndrome is measured \cite{English_2025}. As a consequence, if the transition is unique, the clean critical point is pinned to the self-dual temperature \cite{onsager_crystal_1944,kramers_statistics_1941}.

Second, we show that this self-dual constraint extends from the clean, fully postselected limit to the usual code capacity threshold. For zero-rate \emph{em}-symmetric code families, the principal Boltzmann factor construction~\cite{ohzeki2008duality,PhysRevE.77.061116,PhysRevE.79.021129} constrains the disordered phase boundary above the Nishimori line, yielding an approximate prediction for the optimal code capacity threshold. We extend this framework to mixed Pauli and erasure noise by identifying qubit erasures with bond dilution in the associated statistical mechanical model. This gives a corresponding prediction for the optimal mixed bit-flip and erasure thresholds.

Third, we show that the duality constraint on the threshold applies to concatenated code constructions. In particular, finite-size Kramers-Wannier self-duality is preserved under concatenation of codes which encode a single logical qubit. In the statistical mechanical mapping, optimal decoding of concatenated codes admits a hierarchical formulation, and in the fully postselected limit reduces to an exact real-space renormalization group flow on the associated generalized Ising model on a hierarchical lattice.

Finally, we show that for codes with thresholds that are constrained by duality, the sharper distinctions between code families lie in finite-size performance rather than threshold values. Comparing topological and concatenated growth, we show that the relevant tradeoff is between distance scaling and the entropy of minimum-weight logical operators: topological families have only polynomially many minimum-weight logicals, whereas concatenated families generically have exponentially many (in the code distance). This produces a finite-size crossover in physical overhead, rather than a universal threshold advantage for either code family.

\subsection{Historical introduction and context}

The foundation of fault-tolerant quantum computation is quantum error correction (QEC), whereby logical quantum information is encoded redundantly across many physical qubits. Early seminal works by Shor \cite{Shor_95} and Steane \cite{PhysRevLett.77.793,Steane_96b}, along with the development of the stabilizer formalism \cite{Gottesman_97}, laid the theoretical foundations for QEC. The first fault-tolerant construction involved the use of concatenated quantum codes \cite{Aharonov_97, Knill_96, Zalka_97, Knill_98}. Concatenation consists of recursively encoding QEC codes, such that at each level, the physical subspace (of the full system's Hilbert space) of the inner code is identified with the logical subspace of the outer code \cite{Nielsen_12,preskill1998lecture,gottesman2024}. %Later
Other developments towards fault-tolerant quantum computation include topological QEC codes \cite{Kitaev_03,bravyi1998,Dennis_02,Bombin_06}, which reduced the physical qubit overhead and qubit connectivity requirements through geometric locality. More recently, qLDPC codes relax geometric locality in order to overcome the Bravyi-Poulin-Terhal bound \cite{PhysRevLett.104.050503}, while still achieving low degree (bounded, but non-local) qubit connectivity with potentially asymptotically good code parameters, i.e., code rate and distance scaling \cite{10.1145/3519935.3520017,10.1145/3564246.3585101,Leverrier_2022}.

The power of quantum code concatenation was cemented by the quantum threshold theorem \cite{Aharonov_97, Knill_98, Kitaev_03, Dennis_02}. This theorem established that arbitrarily reliable quantum computation is possible if the error rate per physical qubit operation $p$ is below a critical value called the threshold $p_{c}$. For concatenated coding schemes, with $p<p_{c}$, the logical failure rate $\mathbb{P}_{\mathrm{fail}}$ can be arbitrarily suppressed with sufficiently large numbers of concatenation levels, with the overall space-time resource overhead scaling as $\mathcal{O}(\mathrm{poly}[\log(A/\epsilon)]A)$, where $\mathbb{P}_{\mathrm{fail}}\leq \epsilon$ and $A$ defines the necessary resources without noise \cite{PhysRevLett.117.010501,Nielsen_12}. The existence of a nonzero threshold proved that fault-tolerant quantum computation is physically possible, provided physical error rates can be made sufficiently low. Initial rigorous proofs of the threshold theorem relied heavily on the structure of concatenated codes \cite{Aharonov_97,Kitaev_1997,Knill_98,Aliferis_2006}. It was later shown that optimal maximum likelihood decoding (MLD) of concatenated codes can also be performed efficiently using a message-passing algorithm \cite{PhysRevA.74.052333}. Recently, concatenated code schemes have experienced a resurgence in the context of constant-overhead quantum computation \cite{Yamasaki2024,Yoshida2025,gidney2025}. The performance of concatenated codes can also be analyzed through a mapping to statistical mechanical models on trees \cite{Yadavalli2025,PhysRevResearch.7.023086}.

Despite their theoretical importance, concatenated codes often suffer from large resource overheads and qubit connectivity requirements that increase with concatenation level \cite{preskill1998lecture,cao2025,xbzn-vn37}. Although the stabilizer groups of concatenated codes typically contain non-local generators, fault-tolerant computation with such codes can still be implemented using local gates \cite{Gottesman2000,litinski2025,xbzn-vn37}. Nevertheless, these considerations spurred the development of alternative QEC strategies, most notably topological quantum codes \cite{Kitaev_03, Dennis_02, Bombin_06}.
Topological codes address the issues of large resource overheads and qubit connectivity requirements by combining a qLDPC structure, in which each stabilizer generator has bounded weight and each qubit participates in a bounded number of checks, with geometric locality, so that these checks can be measured using local interactions on a 2D or higher dimensional layout.
Topological codes encode information in global degrees of freedom of a many-body quantum system. Some examples of topological codes are the well-known surface codes \cite{Kitaev_03, bravyi1998, Fowler_12} and color codes \cite{Bombin_06,PhysRevLett.103.090501}.
This combination of local checks, bounded qubit degree, and spatially local syndrome extraction makes topological codes a natural fit for hardware architectures with predominantly nearest-neighbour interactions \cite{Bomb_n_2013}.

The structure of topological codes was later discovered to admit exact mappings to well-studied disordered classical statistical models, in which the threshold of the code separates an ordered (correctable) and disordered (uncorrectable) phase \cite{Dennis_02,Wang_03,Chubb_2021}. These statistical mechanical mappings have motivated efficient optimal decoding algorithms \cite{PhysRevA.90.032326, PhysRevLett.134.190603}. The thresholds of generic codes with translationally invariant, geometrically local stabilizers can also be analyzed through different information measures as in Refs.~\cite{PhysRevB.110.085158,lyons2024}.

More generally, qLDPC codes are QEC code families in which the stabilizer generator weight and qubit degree remain bounded as the code size grows \cite{PRXQuantum.2.040101}. By allowing stabilizer generators to act on qubits non-locally, qLDPC codes can substantially improve encoding efficiency \cite{PRXQuantum.2.040101,PhysRevLett.104.050503,Baspin2022connectivity}. Recent progress has made this direction especially compelling: explicit product-based constructions have improved achievable rate-distance tradeoffs \cite{Tillich_2014,9490244}, asymptotically good qLDPC families are now known \cite{10.1145/3519935.3520017,10.1145/3519935.3520024,Leverrier_2022}, and recent bivariate bicycle (BB) code proposals have shown that qLDPC families can also be competitive at practically relevant finite sizes and circuit-level noise rates \cite{Bravyi2024}. Among the qLDPC landscape, a particularly relevant class for the present work is given by A2BGA codes \cite{PhysRevA.109.022407} and related BB code constructions \cite{Bravyi2024}, which provide CSS families with structure that allows the results developed in this work to be applied.

\subsection{Organization of the manuscript}

The remainder of the manuscript is structured as follows.
In Section~\ref{sec:background} we review the statistical mechanical mapping of QEC, introduce the generalized Kramers-Wannier duality, and summarize how principal Boltzmann factor arguments constrain the phase boundary of disordered models.
In Section~\ref{sec:symmetry} we formulate a finite-size notion of Kramers-Wannier self-duality for CSS codes, show that zero-rate \emph{em}-symmetric families are self-dual in the thermodynamic limit, and prove that generalized Kramers-Wannier self-duality is preserved under concatenation. Together, these results explain the observed near equivalence of phase diagrams across the topological, concatenated and qLDPC families we study.
In Section~\ref{sec:concat} we first show that optimal decoding of concatenated codes admits a natural interpretation as an exact real-space renormalization group flow on the associated hierarchical statistical mechanical model, and we connect this picture to Poulin's efficient message-passing decoder \cite{PhysRevA.74.052333}. We then compare the optimal performance of topological and concatenated codes, and explain the differences observed through an entropic  argument.
In Section~\ref{sec:BB_codes}, we describe the conditions under which more general qLDPC codes, specifically A2BGA codes, are also expected to have their thresholds constrained by self-dual arguments. 
We conclude in Section~\ref{sec:discussion} by outlining directions for future work.

\section{Background}\label{sec:background}

In this section, we summarize the formalism and results from the literature required to prove our main result. First, we briefly review the statistical mechanical mapping of decoding Pauli stabilizer codes \cite{Dennis_02,Wang_03,Chubb_2021}. We then introduce the generalized Kramers-Wannier duality of generalized Ising models \cite{10.1063/1.1665530}, including the effects of electric and magnetic insertions \cite{PhysRevA.97.062320}, and define generalized Kramers-Wannier self-duality. Finally, we introduce the replica framework to predict phase diagrams of self-dual generalized Ising models \cite{ohzeki2008duality,PhysRevE.79.021129,PhysRevE.77.061116,Takeda_2005}.

To speak meaningfully about thresholds, one must consider a sequence of codes $\{\mathcal{C}_{t}\}_{t\in\mathbb{N}}$ whose block length $n_{t}$ tends to infinity. A necessary condition for a nonzero asymptotic threshold is that the code distance $d_{t}$ also diverges with block length \cite{PhysRevLett.115.050502}. Otherwise, the family can correct only a bounded number of errors, and the logical failure probability cannot vanish at any fixed nonzero physical error rate. Families of qLDPC codes with distances that grow either as a power law (with positive exponent) or logarithmic in the block length admit finite thresholds under suitable noise models~\cite{PhysRevA.87.020304}.

Threshold values depend on the noise model under consideration~\cite{PhysRevA.94.042338,PhysRevResearch.6.L042014}.
Noise processes affecting the physical qubits can be modeled in different ways, including coherent and incoherent noise models.
Coherent noise models represent physical noise through a unitary channel~\cite{Iverson_2020}, and may require specific modeling assumptions or noise-tailoring procedures~\cite{PhysRevLett.131.060603,PhysRevA.94.052325,PhysRevLett.121.250502}.
Noise can also be modeled through stochastic Pauli error channels acting on the physical qubits \cite{PhysRevLett.121.190501}, and such incoherent channels can be obtained from coherent channels via the Pauli twirling approximation \cite{Katabarwa2015}.
Incoherent noise models are typically classified into three settings, depending on the locations where potential faults are assigned~\cite{PhysRevResearch.6.L042014}.
In the code capacity setting, errors occur only on the data qubits and syndrome measurements are assumed to be perfect~\cite{Dennis_02}. Phenomenological noise models include both data qubit errors and noisy syndrome measurements \cite{Wang_03}, while circuit-level noise models assign faults to the elementary operations used for state preparation, gates, idle periods and measurements~\cite{Fowler_12}. These increasingly realistic models typically reduce the threshold compared to the code capacity setting. For this reason, code capacity thresholds are best understood as idealized benchmarks that provide a useful first approximation for comparing code families, rather than direct predictions of hardware-level performance.

The threshold, however, captures only the location of the phase transition \cite{Dennis_02,Chubb_2021}. Distinct code families may share the same threshold \cite{PhysRevX.2.021004,Katzgraber_2013} while exhibiting markedly different logical failure suppression below threshold. With experimental platforms now realizing sub-threshold physical error rates \cite{Acharya2025,rqkg-dw31, Bluvstein2026,computing2026}, recent literature has emphasized the importance of sub-threshold scaling of QEC codes \cite{Acharya2025,Lacroix2025,Eickbusch2025}.

The threshold $p_{c}$ admits a natural interpretation within the statistical mechanical mapping of decoding, which we now briefly review.

\subsection{Statistical mechanics mapping of QEC}

The statistical mechanical mapping of Pauli stabilizer codes represents the relative likelihood of different logical cosets as a ratio of partition functions in different sectors of an associated generalized Ising model \cite{Dennis_02,Chubb_2021}. Under this mapping, optimal decoding corresponds to selecting the logical sector with minimal free energy, and the decoding threshold is identified with an order-disorder transition of the model. 
The free energy cost of inserting a nontrivial logical operator provides a useful order parameter to characterize this phase transition \cite{Chubb_2021}. In this work, we restrict our attention to code capacity bit-flip noise for analytic tractability. However, the mapping and analysis used can be extended to much broader classes of noise channels \cite{PhysRevX.2.021004,Chubb_2021,rispler2024}.

Under this mapping, we associate an Ising spin $S_{i}$ to each stabilizer generator $\hat{S}_{i}$. Each physical qubit $b$ becomes an interaction term (or ``bond'') in the Hamiltonian that couples together the spins associated to stabilizers that have support on the qubit $b$. Summing over all spin configurations then enumerates all physical error representatives related by stabilizers. Upon imposing the Nishimori condition \cite{Nishimori_1981}, which fixes a relationship between the disorder probability and the temperature of the model, the partition function computes the cumulative probability of all physical errors within a given logical coset (consistent with the measured syndrome), as required for optimal maximum-likelihood decoding \cite{Chubb_2021}. Adopting the notation of Kovalev \emph{et al.}~\cite{PhysRevA.97.062320}, the resulting Hamiltonian takes the form
\begin{equation}\label{eq:GIM}
    \mathcal{H}=-\sum_{b=1}^{n}J\left[\prod_{j}S_{j}^{\theta_{jb}} \right].
\end{equation}
Here, $b$ is the index of the $n$ physical qubits, $j$ runs over all classical spins $S_{j}$, and $\theta_{jb}=1$ if the stabilizer $\hat{S}_{j}$ acts nontrivially on qubit $b$ and $\theta_{jb}=0$ otherwise. In the QEC setting, $\theta$ therefore coincides with the parity-check matrix. In some formulations of the statistical mechanical mapping, $\theta$ is instead identified with the generator matrix of the code \cite{PhysRevA.97.062320}. However, as we show below, the parity-check representation is more convenient for the analytic results developed here.

To simplify later analysis, we introduce the dimensionless coupling constant $K$, which is equal to the product of the inverse temperature $\beta$ and the coupling constant $J$. For partition functions to encode the cumulative probability of logical cosets, $K$ must be related to the bit-flip probability through the Nishimori conditions, which state that
\begin{equation}\label{eq:nishimori_conditions}
    K_{p}\defeq\beta J=\frac{1}{2}\log\left(\frac{1-p}{p}\right),
\end{equation}
where $p$ is the i.i.d.\ probability of bit-flip noise.

We next introduce electric and magnetic insertions to the system, which are defined through the partition function of such Hamiltonians (following Ref.~\cite{PhysRevA.97.062320})
\begin{align}\label{eq:generalized_partition}
    Z_{\vec{e},\vec{m}}(\theta;K_{p})=&\sum_{S_{i}=\pm 1}\prod_{b=1}^{n}\left[\prod_{j}S_{j}^{\theta_{jb}}\right]^{m_{b}}\times\nonumber\\
    &\exp\left(K_{p}(-1)^{e_{b}}\prod_{j}S_{j}^{\theta_{jb}}\right)
\end{align}
where $\vec{e}$ and $\vec{m}$ are binary vectors specifying electric and magnetic insertions, respectively. The index $b$ runs over all physical qubits (interaction terms in the Hamiltonian), and $j$ runs over all spins (stabilizers). When $e_{b}=1$, the term obtained from the physical qubit $b$ has antiferromagnetic coupling and when $e_{b}=0$, the coupling is ferromagnetic. In other words, if a bit-flip $\hat{\sigma}_{x}$ acts on qubit $b$ we set $e_{b}=1$ and if an identity $\hat{I}$ acts, $e_{b}=0$. When $m_{b}=1$, we bring down the term from the exponential, obtaining an expression proportional to a multi-spin correlation function. Setting $\vec{m}=0$ yields a partition function with electric insertions only, and this partition function $Z_{\vec{e},\vec{0}}(\theta;K_{p})$ is proportional to the cumulative probability of physical errors equivalent up to stabilizers within the logical sector encoded by $\vec{e}$. $Z_{\vec{0},\vec{0}}(\theta;K_{p})$ defines the clean ferromagnetic limit of the model with dimensionless coupling constant $K_{p}$.

We define the threshold as the physical error rate above which the free energy cost of every nontrivial logical coset vanishes in the thermodynamic limit, and below which it diverges \cite{Chubb_2021}. Equivalently, above threshold all logical sectors become asymptotically equally likely. With bit-flip noise, we only need to consider Paulis $\hat{E}\in\{\hat{I},\hat{\sigma}_{x}\}^{\otimes n}$. We denote $\vec{E}$ as the indicator vector of the support of $\hat{E}$, i.e. $E_{b}=1$ iff bond $b$ lies in $\mathrm{supp}(\hat{E})$. We now define the quenched average free energy cost of a nontrivial logical operator $\hat{\mathcal{L}}_{a}\in\mathcal{P}^{\otimes n}$ through
\begin{equation}
    \langle\Delta F_{\hat{\mathcal{L}}_{a}}\rangle_{\hat{E}}=\left\langle-\frac{J}{K_{p}}\log\left(\frac{Z_{\vec{E\mathcal{L}_{a}},\vec{0}}(\theta;K_{p})}{Z_{\vec{E},\vec{0}}(\theta;K_{p})}\right)\right\rangle_{\hat{E}},
\end{equation}
where $\langle.\rangle_{\hat{E}}$ denotes a quenched average over the Paulis $\hat{E}\in\mathcal{P}^{\otimes n}$. Finally, assuming the transition is unique, we define the threshold of the physical error rate $p_{c}$ through
\begin{equation}\label{eq:threshold}
    \lim_{n\to\infty}\langle\Delta F_{\vec{\mathcal{L}}_{a}}\rangle_{\hat{E}}=
\begin{cases}
    \infty & \text{if } p < p_{c} \\
    0  & \text{if } p > p_{c}.
\end{cases}
\end{equation}

\subsection{Generalized Kramers-Wannier duality}

We now review the generalized Kramers-Wannier duality for generalized Ising models of the form in Eq.~\eqref{eq:GIM}~\cite{10.1063/1.1665530}. The primary object on which this duality acts is the interaction matrix $\theta$, whose entries specify which Ising spins participate in which interaction terms. We consider the statistical mechanical model defined by the Hamiltonian specified by $\theta$ at some inverse temperature $\beta$. Applying the duality to this matrix produces a dual interaction matrix $\theta^{*}$, whose associated model exchanges electric and magnetic insertions at a dual inverse temperature $\beta^{*}$.

Given a generalized Ising model which possesses $N_{b}$ interaction terms (each associated with a physical qubit) and $N_{s}$ spins (each associated with a stabilizer generator), and is operating at inverse temperature $\beta$ such that $K=J\beta$, we can define a dual model by its interaction matrix $\theta^{*}$ if it satisfies the following three criteria:
\begin{enumerate}
    \item \textbf{Closure}
\begin{align}
    \theta\left(\theta^{*}\right)^{T}=0
\end{align}
where multiplication is taken over $\mathbb{F}_{2}$.

    \item \textbf{Completeness}
    \begin{equation}
    N_{\theta} + N_{\theta}^{*}=N_{b}
\end{equation}
where $N_{\theta}$ ($N_{\theta}^{*}$) is the rank of $\theta$ ($\theta^{*}$).
\item \textbf{Dual temperature}
\begin{equation}\label{eq:dual_temp}
    \tanh K=e^{-2K^{*}},
\end{equation}
or equivalently $\sinh(2K)\sinh(2K^{*})=1$, where $K^{*}$ is the coupling constant dual to $K$.
\end{enumerate}

The primary model has a partition function $Z_{\vec{0},\vec{0}}(\theta;K)$, and the dual model has $Z_{\vec{0},\vec{0}}(\theta^{*};K^{*})$. If these three criteria are satisfied, then $Z_{\vec{0},\vec{0}}(\theta;K)$ can be mapped to $Z_{\vec{0},\vec{0}}(\theta^{*};K^{*})$. To make this relationship clearer, we first define the symmetric partition functions, which remove temperature-dependent normalization factors. The symmetric partition functions are defined through \cite{10.1063/1.1665530}
\begin{equation}
    Y_{\vec{e},\vec{m}}(\theta;K)=Z_{\vec{e},\vec{m}}(\theta;K)2^{-(N_{s}+N_{g})/2}[\cosh(2K)]^{-N_{b}/2},
\end{equation}
where $N_{s}$ is the number of spins, $N_{b}$ is the number of interaction terms, and the ground state is $2^{N_{g}}$-fold degenerate (with $N_{g}=N_{s}-N_{\theta}$). $Y(\theta^{*};K^{*})$ is defined similarly. If the three criteria above hold, Wegner showed that these models obey
\begin{equation}\label{eq:Wegner_equality}
    Y_{\vec{0},\vec{0}}(\theta;K)=Y_{\vec{0},\vec{0}}(\theta^{*};K^{*}).
\end{equation}
More generally, electric and magnetic defects are exchanged under the duality. Inserting both types simultaneously produces a phase factor determined by their dot product taken over $\mathbb{F}_{2}$ \cite{PhysRevA.97.062320}. These symmetric partition functions are related through
\begin{equation}
    Y_{\vec{e},\vec{m}}(\theta;K)=(-1)^{\vec{e}\cdot\vec{m}}Y_{\vec{m},\vec{e}}(\theta^{*};K^{*}).
\end{equation}

We say that the model parametrized by $\theta$ is self-dual if the bulk free energy densities of the primary and dual models coincide in the thermodynamic limit within the trivial sector $(\vec e,\vec m)=(\vec0,\vec0)$. That is, the model is self-dual if
\begin{align}\label{eq:self_dual_1}
    f(K)&\defeq \lim_{n\to\infty}-\frac{1}{n}\log(Y_{\vec{0},\vec{0}}(\theta;K))\nonumber\\
    &= \lim_{n\to\infty}-\frac{1}{n}\log(Y_{\vec{0},\vec{0}}(\theta^{*};K))\eqdef f^{*}(K).
\end{align}
Equivalently, using Eq.~\eqref{eq:Wegner_equality}, this condition can be rewritten as
\begin{equation}\label{eq:self_dual_2}
    f(K)=f(K^{*}),
\end{equation}
where we have also used that generalized Kramers-Wannier duality is idempotent, that is, $K^{**}=K$ and $\theta^{**}=\theta$. Finally, if the critical point is unique, self-duality pins it to the self-dual temperature defined by $K^{*}=K=K_{c}$. Using Eq.~\eqref{eq:dual_temp}, this yields
\begin{equation}\label{eq:K_c}
    K_{c}=\frac{1}{2}\log(\sqrt2+1).
\end{equation}

\subsection{Replica analysis of Kramers-Wannier self-dual codes}\label{subsec:replica_analysis}

For self-dual generalized Ising models, the phase boundary can be approximated to leading order through the principal Boltzmann factor construction using the replica trick \cite{ohzeki2008duality,PhysRevE.77.061116,PhysRevE.79.021129}. We summarize the argument here in the notation of the generalized Ising model introduced above. First we assume that each physical qubit is acted upon by i.i.d.\ bit-flip noise, such that for all physical qubits indexed by $b$, we have
\[
\mathbb{P}(e_{b}=1)=p,\qquad \mathbb{P}(e_{b}=0)=1-p.
\]
For a fixed disorder realization $\vec{e}$, the resulting bond-disordered model is described by the generalized partition function $Z_{\vec e,\vec 0}(\theta;K_{p})$, see Eq.~\eqref{eq:generalized_partition}. Nontrivial logical sectors are then represented by shifted electric insertions $Z_{\vec{e\mathcal{L}}_{a},\vec 0}(K_{p})$, where $\vec{e\mathcal{L}}_{a}$ is the binary representative of the product of the Pauli $\hat{e}$ and the nontrivial logical operator $\hat{\mathcal{L}}_{a}$.

For $N_{\mathrm{rep}}$ replicas of the disordered model, the replicated Boltzmann weight associated with a single interaction term depends only on the number $l\in\{0,\dots,N_{\mathrm{rep}}\}$ of replicas in which the local interaction variable
\[
\prod_{j}(S_{j}^{\theta_{jb}})_{\alpha}
\]
is equal to $-1$. Here $\alpha$ indicates the replica index which runs from $1$ to $N_{\mathrm{rep}}$. We now fix $l$ replicas to have their local interaction variable equal to $-1$ and $N_{\mathrm{rep}}-l$ replicas having this variable equal to $+1$. Averaging over the disorder then gives the local Boltzmann factor
\begin{align}
x_{l}(K,p)&=\left\langle\exp\left(\sum_{\alpha=1}^{N_\mathrm{rep}}K(-1)^{e_{b}}\prod_{j}(S_{j}^{\theta_{jb}})_{\alpha}\right)\right\rangle_{e_{b}}
\nonumber\\
&=(1-p)e^{(N_{\mathrm{rep}}-2l)K}+pe^{-(N_{\mathrm{rep}}-2l)K},
\label{eq:xk}
\end{align}
where we have considered the local Boltzmann factor at some fixed interaction term $b$. Here, we do not assume that the Nishimori conditions hold. Instead, we treat $K$ and $p$ as independent variables. This lets us obtain the replica estimate for the entire phase boundary, rather than just the critical point along the Nishimori line (where $K=K_{p}$).

The corresponding dual local Boltzmann factors are obtained via the Fourier transform of Eq.~\eqref{eq:xk}, yielding
\begin{align}
x_l^{*}(K,p)&=2^{-N_{\mathrm{rep}}/2}\left[(1-p)+(-1)^l p\right]\times\nonumber\\
&(e^{K}+e^{-K})^{N_{\mathrm{rep}}-l}(e^{K}-e^{-K})^{l}.
\label{eq:xstark}
\end{align}

The principal Boltzmann factor approximation retains only the $l=0$ sector, corresponding to the case in which all replicated local interaction variables are equal to $+1$. One then imposes
\begin{equation}
x_{0}(K,p)=x_{0}^{*}(K,p),
\end{equation}
and finally takes the replica limit $N_{\mathrm{rep}}\to 0$. Since both sides equal $1$ at $N_{\mathrm{rep}}=0$, the leading-order condition is obtained by matching the derivatives of their logarithms at $N_{\mathrm{rep}}=0$:
\begin{align}
\left.\frac{\partial}{\partial N_{\mathrm{rep}}}\log x_{0}(K,p)\right|_{N_{\mathrm{rep}}=0}
&=K(1-2p),\\
\left.\frac{\partial}{\partial N_{\mathrm{rep}}}\log x_{0}^{*}(K,p)\right|_{N_{\mathrm{rep}}=0}
&=\log(2\cosh K)-\frac{1}{2}\log 2.
\end{align}
Equating these expressions gives the implicit curve
\begin{equation}\label{eq:phase_boundary}
F(p,K)\defeq K(1-2p)-\left[\log(2\cosh K)-\frac{1}{2}\log 2\right]=0.
\end{equation}
For self-dual models, this curve provides the leading-order replica prediction for the phase boundary~\cite{ohzeki2008duality,PhysRevE.77.061116,PhysRevE.79.021129}. Restricting Eq.~\eqref{eq:phase_boundary} to the Nishimori line, see Eq.~\eqref{eq:nishimori_conditions},  gives
\begin{equation}
-p\log p-(1-p)\log(1-p)=\frac{1}{2}\log 2,
\end{equation}
or equivalently
\begin{equation}
H_{2}(p)=\frac{1}{2},
\end{equation}
where $H_{2}$ is the binary entropy function. Thus, the principal Boltzmann factor construction predicts the intersection of the phase boundary with the Nishimori line, which we call the Nishimori point, at $p_{N}\approx 0.110028$~\cite{Takeda_2005}.

\section{Symmetry constraints and phase diagram equivalence}\label{sec:symmetry}

In this section, we show that generalized Kramers-Wannier self-duality imposes strong symmetry constraints on the statistical mechanical models associated with \emph{em}-symmetric CSS codes. We first reformulate Wegner's closure and completeness conditions in the language of chain complexes, making explicit how logical coset probabilities appear as partition functions with electric insertions. This naturally leads to a notion of Kramers-Wannier self-duality with logical sector mixing for finite-size codes, and we prove that \emph{em}-symmetric CSS codes satisfy this relation at finite size. We then show that, for code families with a unique threshold, the thermodynamic limit is self-dual if and only if the asymptotic code rate vanishes. As a consequence, the principal Boltzmann factor construction predicts a common zero-rate threshold and constrains the phase boundary above the Nishimori line for this broad class of codes. We further show that this self-duality is preserved under the self-concatenation map of suitable $[[n,1,d]]$ seed codes, which motivates the hierarchical decoding formulation developed later. Finally, we extend the same framework to mixed Pauli-erasure noise, yielding a corresponding three-dimensional phase diagram.

We first re-express the closure and completeness conditions of Wegner's duality in the language of a chain complex over $\mathbb{F}_{2}$. Such an expression of Wegner's duality has been formulated as a chain complex previously, such as in Ref.~\cite{10.1063/1.5039735}. However, our formulation differs in assigning only the parity-check matrices to the boundary maps (and not the generator matrices) with explicit inclusion of logical representatives. Here, we make manifest the encoding of logical coset probabilities as the partition functions with electric insertions as explained below.

Let $C_{1}\simeq\mathbb{F}_{2}^{N_{b}}$ represent the interaction terms (identified with physical qubits in the QEC mapping), and $C_{2}\simeq \mathbb{F}_{2}^{N_{s}}$ represent the Ising spins (equivalently, the $Z$-type stabilizers), so that the parity-check matrix is $H_{Z}=\theta$. Define boundary maps
\begin{equation}
    \partial_{2}\defeq \theta^{\top}\,:\,C_{2}\to C_{1},\qquad \partial_{1}\defeq\theta^{*}\,:\,C_{1}\to C_{0},\qquad
\end{equation}
with $C_{0}\simeq\mathbb{F}_{2}^{N_{s}^{*}}$ the dual-spin space (equivalently the $X$-type stabilizers). Wegner's closure condition $\theta(\theta^{*})^{\top}=0$ is equivalent to $\partial_{1}\partial_{2}=\theta^{*}\theta^{\top}=0$, hence $\mathrm{im}(\partial_{2})\subseteq \mathrm{ker}(\partial_{1})$. This closure is always satisfied by the chain complex condition. Wegner's completeness relation $N_{\theta}+N_{\theta}^{*}=N_{b}=\mathrm{dim}(C_{1})$, together with this inclusion, implies exactness at $C_{1}$, i.e., $\mathrm{ker}(\partial_{1})=\mathrm{im}(\partial_{2})$. This chain complex is indicated below:
\begin{center}
    \begin{tikzpicture}
  \matrix (M) [matrix of math nodes, column sep=1cm, row sep=1cm, nodes={anchor=center}] {
    C_{2} & C_{1} & C_{0}\\
  };
  \draw[->] (M-1-1) -- node[above] {$\partial_2=\theta^{\top}$} (M-1-2);
  \draw[->] (M-1-2) -- node[above] {$\partial_1=\theta^{*}$} (M-1-3);
\end{tikzpicture}
\end{center}
For CSS codes with $k>0$, identifying $\theta^{*}$ with $H_{X}$ violates the strict completeness relation by $k$ (the number of logical qubits). This failure is precisely the rank of the first homology group $H_{1}$. We can complete this map by including $k$ independent representatives of the $X$-type logical operators,
\begin{equation}
    \partial_{1}=\theta^{*}=\begin{bmatrix}
        H_{X}\\L_{X}
    \end{bmatrix},
\end{equation}
thereby satisfying all of Wegner's conditions. An independent corresponding chain complex can similarly be defined for the $Z$-type stabilizers and logical operators.

Given a Pauli stabilizer code which encodes $k>0$ logical qubits, for each measured syndrome, one obtains a vector of partition functions whose elements encode the different logical coset probabilities. Here, we restrict to bit-flip noise for clarity, but one can extend this mapping to phase-flip or depolarizing noise. Let the physical qubit (bond) space be $C_{1}\simeq \mathbb{F}_{2}^{N_{b}}$ and assume the code encodes $k>0$ logical qubits. Choose $k$ independent representatives $\{\hat{X}_a\}_{a=1}^{k}$ of the logical $X$ operators, and let $\vec{L}^{(a)}\in \mathbb{F}_{2}^{N_{b}}$ denote the indicator vector of the support of $\hat{X}_a$. Collect these representatives into the matrix,
\begin{equation}
    L_{X}\in\mathbb{F}_{2}^{k\times N_{b}},
    \qquad (L_X)_{a,b}\defeq L^{(a)}_{b}.
\end{equation}
For any logical label $\lambda\in\mathbb{F}_{2}^{k}$, define the corresponding electric insertion pattern,
\begin{equation}
    \vec{e}(\lambda)\defeq \lambda^{\top}L_{X}\in\mathbb{F}_{2}^{N_{b}},
\end{equation}
so that $\vec{e}(\lambda)$ specifies which interaction terms have their couplings flipped. For the clean model, the partition function in logical sector $\lambda$ is then
\begin{equation}\label{eq:Z_lambda_def}
    Z_{\lambda}(\theta;K_{p})\defeq Z_{\vec{e}(\lambda),\vec{0}}(\theta;K_{p}),
\end{equation}
where $Z_{\vec{e},\vec{m}}(\theta;K_{p})$ denotes the generalized partition function (see Eq.~\eqref{eq:generalized_partition}) with electric insertion $\vec{e}$ and magnetic insertion $\vec{m}$ (with $\vec{m}=\vec{0}$ here).

Once again, to remove temperature-dependent normalization factors, we define the corresponding symmetric partition functions by
\begin{equation}\label{eq:Y_lambda_def}
    Y_{\lambda}(\theta;K_{p})\defeq 2^{-(N_{s}+N_{g})/2}[\cosh(2K_{p})]^{-N_{b}/2}
    Z_{\lambda}(\theta;K_{p}),
\end{equation}
and collect them into the $2^{k}$-component vector
\begin{equation}\label{eq:Yvec_def}
    \vec{Y}(\theta;K_{p})\defeq\big(Y_{\lambda}(\theta;K_{p})\big)_{\lambda\in\mathbb{F}_{2}^{k}}
    \in \mathbb{R}^{2^{k}},
\end{equation}
where the components are ordered by a fixed convention (e.g., lexicographic
order on $\lambda$). For $k=1$, this reduces to $\vec Y(\theta;K_{p})=(Y_{0}(\theta;K_{p}),Y_{1}(\theta;K_{p}))^{\top}$. To make this chain complex construction clear, let us now consider a simple example.

\subsection{Example: Steane code}

We now construct the relevant chain complex for the Steane code with bit-flip noise in the code capacity setting. First, we write down the ($Z$-type) parity-check matrix, which is equivalent to that of the $[7,4]$ Hamming code, and we set this equal to the interaction matrix $\theta$:
\[
\theta= H_{Z}=\begin{bmatrix}
    1&1&0&1&1&0&0\\
    1&0&1&1&0&1&0\\
    0&1&1&1&0&0&1
\end{bmatrix}.
\]
Each row represents a stabilizer, and each column represents a physical qubit. Each stabilizer acts on four qubits, hence each row has four nonzero elements. This model has $N_{s}=3$ classical Ising spins (stabilizers), and $N_{b}=7$ interaction terms (physical qubits). To obtain the dual model, we must include a logical representative, which we can take to be a $\hat{\sigma}_{x}$ on every physical qubit, since the Steane code permits transversal logical Paulis. Then, incorporating this row, we obtain
\[
\theta^{*}=\begin{bmatrix}
    H_{X}\\
    L_{X}
\end{bmatrix}=\begin{bmatrix}
    1&1&0&1&1&0&0\\
    1&0&1&1&0&1&0\\
    0&1&1&1&0&0&1\\
    1&1&1&1&1&1&1
\end{bmatrix}.
\]

The model parametrized by $\theta$ has $N_{g}=N_{s}-N_{\theta}=0$, so each sector Hamiltonian has a unique ground state. The dual model parametrized by $\theta^{*}$ similarly has $N^{*}_{g}=0$.

\subsection{Phase diagram constraints}\label{subsec:phase_diagram_constraints}

Here we show how the phase diagrams of \emph{em}-symmetric codes are constrained by generalized Kramers-Wannier duality. Following the conventions above, we first define a finite-size notion of Kramers-Wannier self-duality of the clean model through the following.

\begin{definition}[Finite-size Kramers-Wannier self-dual code]
A CSS Pauli stabilizer code is said to be finite-size Kramers-Wannier self-dual if there exists a choice of primary and dual logical-sector bases such that
\begin{equation}
    \vec{Y}(\theta;K_{p}^{*})=\mathcal{H}_{k}\vec{Y}(\theta;K_{p}),
\end{equation}
for all $K_{p}\in[0,\infty)$, where $K_{p}^{*}$ is the dual coupling defined by
\begin{equation}
    e^{-2K_{p}^{*}}=\tanh K_{p},
\end{equation}
and $\mathcal{H}_{k}$ is the normalized Hadamard matrix
\begin{equation}
    (\mathcal{H}_{k})_{\vec{\mu},\vec{\lambda}}=2^{-k/2}(-1)^{\vec{\mu}\cdot\vec{\lambda}},
    \qquad \vec{\mu},\vec{\lambda}\in\mathbb{F}_{2}^{k}.
\end{equation}
Here, $\vec{Y}(\theta;K_{p})$ and $\vec{Y}(\theta;K_{p}^{*})$ denote the vectors of symmetric partition functions of the logical sectors of the primary model at the primary and dual temperatures, respectively, ordered according to some fixed convention.
\end{definition}

We demonstrate why this is a useful definition below, but first we provide motivation by showing that \emph{em}-symmetric CSS codes obey finite-size Kramers-Wannier self-duality. We call a CSS code \emph{em}-symmetric if the parity-check matrices $H_{X}$ and $H_{Z}$ are related through row and column permutations, that is, $H_{X}=P_{1}H_{Z}P_{2}$ for some permutation matrices $P_{1},P_{2}$. In the present work, we do not require translation invariance, locality, or an underlying lattice geometry. This encompasses a broad range of codes: from manifestly self-dual codes (where $H_X$ and $H_Z$ coincide without any relabeling, e.g., many 2D color code realizations \cite{Bombin_06}), to codes whose $X$- and $Z$-stabilizers are the same local patterns placed on different sublattices. A canonical example is the surface code, where vertex and plaquette-type stabilizers are translations of one another under a lattice symmetry (possibly composed with a dual-lattice identification), so that a spatial translation of qubits together with a reindexing of stabilizers implements the permutation equivalence above.

Observe that for the clean 2D Ising model on a torus (to which a fully postselected toric code maps, with $k=2$ logical qubits), the logical mixing matrix is equal to the normalized order-$4$ Hadamard matrix after a compatible choice of logical sector bases \cite{english2025_donut,PhysRevB.55.11045}. This model can be readily seen to be finite-size Kramers-Wannier self-dual. For codes with $k=1$ logical qubit, we obtain a mixing matrix proportional to the order-$2$ Hadamard matrix.

\begin{theorem}\label{thm:1}
Let $C$ be a CSS code on $n=N_{b}$ physical qubits with parity-check matrices $H_Z,H_X$, and assume that
\[
H_X=P_1H_ZP_2,
\]
where $P_{1}$ and $P_{2}$ are permutation matrices. Let
\[
\theta:=H_Z,\qquad \theta^*:=\begin{bmatrix} H_X \\ L_X \end{bmatrix},
\]
where the rows of $L_X$ are $k$ independent logical $X$ representatives. Then $\theta^*$ is a Wegner dual interaction matrix for $\theta$. Moreover, after a compatible choice of logical sector basis on the dual side,
\[
\vec Y(\theta;K_{p}^*)=\mathcal{H}_{k}\vec{Y}(\theta;K_{p}),\qquad (\mathcal H_k)_{\vec{\mu},\vec{\lambda}}=2^{-k/2}(-1)^{\vec{\mu}\cdot\vec{\lambda}}.
\]
\end{theorem}

\begin{proof}
We break the proof into two steps.

\medskip
\noindent\textbf{1. $\theta^*$ satisfies Wegner closure and completeness.}

By construction $\theta=H_Z$. Since $C$ is CSS, the stabilizers of the code are split into $X$- and $Z$-type, and can be expressed through their parity-check matrices $H_{X}$ and $H_{Z}$. Because the stabilizers of QEC codes must commute with one another,
\[
H_ZH_X^{\top}=0 \qquad (\text{over } \mathbb F_2).
\]
Also, every logical $X$ commutes with every $Z$-type stabilizer, hence
\[
H_ZL_X^{\top}=0.
\]
Therefore
\[
\theta(\theta^*)^{\top}
=
H_Z
\begin{bmatrix}
H_X^{\top} & L_X^{\top}
\end{bmatrix}
=
\begin{bmatrix}
H_ZH_X^{\top} & H_ZL_X^{\top}
\end{bmatrix}
=0,
\]
so the closure condition holds.

For completeness, let $r_Z=\rank(H_Z)$ and $r_X=\rank(H_X)$. Since
$H_X=P_1H_ZP_2$ and permutations preserve rank,
\[
r_X=r_Z.
\]
The rows of $L_X$ are chosen to be independent modulo the row space of $H_X$,
so
\[
\rank(\theta^*)=\rank\!\begin{bmatrix}H_X\\L_X\end{bmatrix}=r_X+k.
\]
Using the CSS dimension formula
\[
k=n-r_X-r_Z,
\]
we obtain
\begin{align}
    N_{\theta}+N_{\theta}^{*}&=r_Z+(r_X+k)\nonumber\\
    &=r_Z+r_X+n-r_X-r_Z\nonumber\\
    &=n=N_b.\nonumber
\end{align}
Hence, $\theta^*$ satisfies Wegner completeness, and therefore is a
dual interaction matrix for $\theta$.

\medskip
\noindent\textbf{2. The dual temperature mixes logical sectors.}

When we compute the partition function of the dual model, fixing the logical spins into a configuration consistent with a logical sector $\mu\in\mathbb{F}_{2}^{k}$ sets the remaining part of the partition function into the primary model in the respective logical sector. Then, summing over all logical spin configurations sums over all logical sector partition functions. When we insert a magnetic insertion in a given logical sector $\lambda\in\mathbb{F}_{2}^{k}$ chosen in the basis of $Z$-type logical operators dual to the $X$-type logical representatives included in $L_{X}$, the stabilizer spins cancel (because logicals must commute with stabilizers) and we are left with the logical spin brought down from the exponential as in Eq.~\eqref{eq:generalized_partition}. In doing so, we pick up a minus sign if this logical degree of freedom has odd pairing with the electric sector choice:
\[
Z_{\vec{0},\vec{m}(\lambda)}(\theta^{*};K_{p})=\sum_{\mu\in\mathbb{F}_{2}^{k}}(-1)^{\vec{\mu}\cdot \vec{\lambda}}Z_{\vec{e}(\mu),\vec{0}}(\theta;K_{p}).
\]
We can rewrite the left hand side using the standard electric-magnetic duality~\cite{PhysRevA.97.062320},
\[
Z_{\vec{0},\vec{m}(\lambda)}(\theta^{*};K_{p})=\frac{Z_{\vec{m}(\lambda),\vec{0}}(\theta;K_{p}^{*})}{A(K_{p})},
\]
where $A(K_{p})$ is the temperature-dependent normalization factor. Combining the two gives us
\[
\frac{Z_{\vec{m}(\lambda),\vec{0}}(\theta;K_{p}^{*})}{A(K_{p})}=\sum_{\mu\in\mathbb{F}_{2}^{k}}(-1)^{\vec{\mu}\cdot \vec{\lambda}}Z_{\vec{e}(\mu),\vec{0}}(\theta;K_{p}).
\]
Finally, moving to the symmetric partition functions, we absorb the $A(K_{p})$ factor and pick up a factor of $2^{-k/2}$ as the number of spins differs between the primary and dual model by precisely $k$, see Eq.~\eqref{eq:Y_lambda_def}. Thus, we arrive at
\[
\vec{Y}(\theta;K_{p}^{*})=\mathcal{H}_{k}\vec{Y}(\theta;K_{p}).
\]
\end{proof}

With \emph{em}-symmetric CSS codes obeying finite-size Kramers-Wannier self-duality, under the definition of threshold as in Eq.~\eqref{eq:threshold}, we now show that in the thermodynamic limit, \emph{em}-symmetric CSS codes which possess a unique threshold are Kramers-Wannier self-dual if and only if the (asymptotic) code rate is zero.

\begin{corollary}\label{thm:2}
Let $\{C_n\}_{n\in\mathbb{N}}$ be an $em$-symmetric CSS family with asymptotic rate
\[
R:=\lim_{n\to\infty}\frac{k_n}{n},
\]
and suppose the family has a unique threshold in the sense that below threshold the identity sector asymptotically dominates all nontrivial sectors, while above threshold all logical sectors become asymptotically equiprobable. Then,
\begin{align}
    &\lim_{n\to\infty}\frac{1}{n}\log Y_{\vec{0}}(\theta;K_{p}^*)\nonumber\\
    &=\begin{cases}
\displaystyle
\lim_{n\to\infty}\frac{1}{n}\log Y_{\vec{0}}(\theta;K_{p})-\frac{R}{2}\log 2,
& p<p_c,\\[1.2em]
\displaystyle
\lim_{n\to\infty}\frac{1}{n}\log Y_{\vec{0}}(\theta;K_{p})+\frac{R}{2}\log 2,
& p>p_c.
\end{cases}\nonumber
\end{align}
Consequently, the family is Kramers-Wannier self-dual in the thermodynamic limit if and only if $R=0$.
\end{corollary}
\begin{proof}
From Theorem~\ref{thm:1}, after ordering the logical sectors such that the identity sector corresponds to the first component, the first row of the Hadamard matrix gives
\[
Y_{\vec{0}}(\theta;K_{p}^*)=2^{-k/2}\sum_{\lambda\in\mathbb F_2^k} Y_\lambda(\theta;K_{p}).
\]
Thus,
\begin{align}
    &\lim_{n\to\infty}\frac{1}{n}\log Y_{\vec{0}}(\theta;K_{p}^{*})\nonumber\\
    &=\lim_{n\to\infty}\frac{1}{n}\log\left(2^{-k/2}\sum_{\lambda\in\mathbb F_2^k} Y_\lambda(\theta;K_{p})\right)\nonumber\\
    &= -\frac{R}{2}\log2+\lim_{n\to\infty}\frac{1}{n}\log\left(\sum_{\lambda\in\mathbb F_2^k} Y_\lambda(\theta;K_{p})\right)\nonumber
\end{align}

Now, we use the threshold assumption. By the definition of a threshold in Eq.~\eqref{eq:threshold} (which must hold for the clean model as well since the free energy cost of nontrivial logical operators of the disordered model is upper bounded by the clean model; see Section~III of Ref.~\cite{PhysRevA.97.062320}), the identity sector partition function dominates asymptotically below threshold ($p<p_{c})$, and we obtain
\[
\lim_{n\to\infty}\frac{1}{n}\log Y_{\vec{0}}(\theta;K_{p}^{*})=-\frac{R}{2}\log2+\lim_{n\to\infty}\frac{1}{n}\log Y_{\vec{0}}(\theta;K_{p}).
\]

Above threshold $(p>p_{c})$, all logical sector partition functions become asymptotically equal, and we obtain
\begin{align}
      \lim_{n\to\infty}\frac{1}{n}\log Y_{\vec{0}}(\theta;K_{p}^{*})&=-\frac{R}{2}\log2+\lim_{n\to\infty}\frac{1}{n}\log\left(2^{k}Y_{\vec{0}}(\theta;K_{p})\nonumber\right)\\
      &= \frac{R}{2}\log 2+\lim_{n\to\infty}\frac{1}{n}\log Y_{\vec{0}}(\theta;K_{p}).\nonumber
\end{align}
This proves the two asymptotic relations.

Finally, an \emph{em}-symmetric family is Kramers-Wannier self-dual in the thermodynamic limit precisely when the bulk free energies of the primary and dual models coincide, see Eqs.~\eqref{eq:self_dual_1} and \eqref{eq:self_dual_2}. The only possible extensive discrepancy is the sector mixing term computed above. If $R=0$, then $\frac{k}{n}\to 0$ and this defect vanishes in both phases, so the family is self-dual in the thermodynamic limit. Conversely, if $R>0$, then in the disordered phase the bulk free energies differ by $\frac{R}{2}\log 2$, so the family cannot be thermodynamically self-dual. Hence, thermodynamic Kramers-Wannier self-duality holds if and only if $R=0$.
\end{proof}

With this self-duality relation established for zero-rate \emph{em}-symmetric CSS codes, we now show that the principal Boltzmann factor condition proposed in Refs.~\cite{PhysRevE.80.011141,PhysRevX.2.021004} under a replica trick framework constrains the threshold to that given by the hashing bound with zero rate~\cite{PhysRevA.54.1098,PhysRevA.54.3824}.

\begin{result}
    The threshold obtained by enforcing the principal Boltzmann factor for any $em$-symmetric CSS Pauli stabilizer code (with zero rate) is determined by the hashing bound with zero rate. This threshold is given by the physical error rate $p_{c}$ which satisfies $H_{2}(p_{c})=\frac{1}{2}$, where $H_{2}$ is the binary entropy function. This value is approximately given by $p_{c}^{x}=p_{c}^{z}\approx 0.110028$.
\end{result}

Counterexamples to this zero-rate phenomenology include surface codes on closed hyperbolic manifolds. In this setting, the logical sector structure remains extensive in the thermodynamic limit, and the associated random-bond Ising model is not self-dual even on self-dual lattices~\cite{PhysRevE.107.024125}. Correspondingly, the maximum-likelihood threshold need not be pinned near the zero-rate self-dual value: for the \{5,5\} hyperbolic surface code family, Ref.~\cite{PhysRevE.107.024125} reports \(p_{\mathrm{th}}^{\mathrm{MLD}} \approx 0.0228\). This is consistent with Corollary~\ref{thm:2}, which identifies asymptotically vanishing rate as the condition under which the sector-mixing defect becomes subextensive and thermodynamic self-duality is recovered. 

Above the Nishimori line, the phase boundary $T_{c}(p)$ for $p<p_{N}$ is constrained in the replica limit $N_{\mathrm{rep}}\to 0$ with the principal Boltzmann factor enforced for any \emph{em}-symmetric CSS Pauli stabilizer code with zero rate, where $p_{N}$ is the multicritical Nishimori point. In this limit, the phase boundary is given by Eq.~\eqref{eq:phase_boundary}. We emphasize that this condition should be interpreted as an approximation to the phase boundary only above the Nishimori line. For the random-bond Ising model, the Nishimori point is known to be a multicritical and unstable fixed point of the renormalization group (RG). This means that RG flows originating above the Nishimori line are attracted to the clean Ising fixed point, whereas flows below the line run towards the zero-temperature (spin-glass) fixed point \cite{Picco_2006}. Since the present construction becomes exact at the clean (self-dual) fixed point, it is natural to expect that the resulting branch provides a reasonable approximation to the phase boundary in the region controlled by this fixed point. By contrast, no such control exists in the regime below the Nishimori line, where this approximation is not expected to hold. In particular, while the replica trick predicts an approximate zero-temperature critical point at $p_{c}(T=0)=0$, this is empirically found to be far from the true zero-temperature critical point \cite{PhysRevE.84.040101}.

Concatenation preserves the \emph{em}-symmetry of a CSS seed code. Indeed, if the seed code satisfies $H_X=P_1H_ZP_2$ for some row and column permutation matrices $P_1,P_2$, then replacing each physical qubit by an encoded copy of the same seed code preserves this permutation equivalence at every concatenation level. Moreover, for a $[[n,1,d]]$ seed code, concatenation gives a $[[n^r,1,d^r]]$ code, so the number of logical qubits remains $k=1$. If the seed admits a transversal logical Pauli representative, this representative is inherited by the concatenated family.

\begin{lemma}
Let $\mathcal{C}$ be an \emph{em}-symmetric $[[n,1,d]]$ CSS seed code. There exists a concatenation map such that every finite concatenation level remains finite-size Kramers-Wannier self-dual.
\end{lemma}
\begin{proof}
    If the seed code is \emph{em}-symmetric, then there must exist equal-weight representatives of $X$- and $Z$-type logical operators. Then, by choosing a concatenation map which acts on these equal-weight logical representatives, each $X$-type stabilizer will have a corresponding equal-weight $Z$-type stabilizer. Moreover, we will retain $k=1$ logical qubit. Thus at each level, the code remains \emph{em}-symmetric. Therefore, we can take our dual model as in Theorem~\ref{thm:1}. Hence our code remains finite-size Kramers-Wannier self-dual at any concatenation level.
\end{proof}

Because concatenation of such codes preserves \emph{em}-symmetry, a unique critical point is again pinned to the self-dual $K_{c}=\frac{1}{2}\log(\sqrt{2}+1)$ and the arguments above hold. The principal Boltzmann factor construction does not require translational invariance of the incidence structure $\theta$ as with topological codes with a regular stabilizer structure. It only requires that the disorder distribution factorizes over interaction terms and that each interaction contributes a two-valued ($\pm 1$) local weight. Even though $\theta$ is potentially highly inhomogeneous in concatenated codes, the $l=0$ replica sector contributes $x_{0}^{N_b}$ irrespective of the detailed form of $\theta$, so the leading order principal-factor condition remains $x_{0}(p,K)=x_{0}^{*}(p,K)$ for uniform i.i.d. disorder.

\subsection{Mixed Pauli and erasure channels}

The statistical mapping can be extended to incorporate the effects of qubit erasures straightforwardly. The modifications to Pauli error thresholds in the presence of loss were first investigated in Ref.~\cite{PhysRevLett.102.200501}. Recent work has re-emphasized the potential benefits associated with erasure qubits \cite{PhysRevResearch.7.013249,985g-58gd}, and this builds on a mature literature focused on photonic fault-tolerant quantum computing, where erasures are the dominant source of errors \cite{PhysRevA.52.3489,Knill2001,PhysRevA.68.064303,Bartolucci2023,PhysRevA.72.032307}.

In the statistical mechanical model obtained via the mapping, the quenched disorder represents physical errors, i.e., through the identification of bit-flips with an antiferromagnetic coupling with probability $p$, and ferromagnetic coupling with a complementary probability $1-p$. When a qubit is erased, the erasure location is known, and no quantum information remains accessible at that location. In the statistical mechanical mapping, this is represented by setting the corresponding coupling to zero \cite{d8rx-srpn}. We represent the probability of such a 0 coupling constant by the variable $q$ (the bond dilution; see Ref.~\cite{Picco_2006}), and therefore obtain a 3-dimensional phase diagram in terms of $T,p,q$, where $T\in[0,\infty)$, $p,q\in[0,1]$.

To obtain the replica trick prediction of the phase boundary, we modify the partition function in Eq.~\eqref{eq:generalized_partition} to incorporate this bond dilution represented through
\begin{align}
        Z_{\vec{e},\vec{m}}(\theta;K_{p})=&\sum_{S_{i}=\pm 1}\prod_{b=1}^{n}\left[\prod_{j}S_{j}^{\theta_{jb}}\right]^{m_{b}}\times\nonumber\\
    &\exp\left(K_{p}\eta_{b}(-1)^{e_{b}}\prod_{j}S_{j}^{\theta_{jb}}\right),
\end{align}
where we now impose $\mathbb{P}(\eta_{b}=0)=q$ and $\mathbb{P}(\eta_{b}=1)=1-q$. With this new bond distribution, we obtain the local averaged Boltzmann factor through (see Eq.~\eqref{eq:xk})
\begin{equation}
    x_{l}(q,p,K)=(1-q)\left[(1-p)e^{(N_{\mathrm{rep}}-2l)K}+p e^{-(N_{\mathrm{rep}}-2l)K}\right]+q
\end{equation}
and by the linearity of the Fourier transform, we obtain the dual Boltzmann factor through
\begin{align}
    x_{l}^{*}(q,p,K)=(1-q)x^{*}_{l}(p,K) + q x_{l}^{*}(\mathrm{const}\ 1),
\end{align}
where $\mathrm{const}\ 1$ indicates that the local Boltzmann factor does not depend on the replica spins. As it is the Fourier transform of $x_{l}=1$ for all $l$, we obtain a Kronecker delta at the $l=0$ mode:
\begin{equation}
    x_{l}^{*}(\mathrm{const}\ 1)=2^{N_{\mathrm{rep}}/2}\delta_{l,0}
\end{equation}
and hence the dual Boltzmann factor is obtained through
\begin{align}
    x_{l}^{*}(q,p,K)=&(1-q)\big[2^{-N_{\mathrm{rep}}/2}((-1)^{l}p + (1-p))\times\nonumber\\    
    &(e^{K}+e^{-K})^{N_{\mathrm{rep}}-l}(e^{K}-e^{-K})^{l} \big] + q 2^{N_{\mathrm{rep}}/2}\delta_{l,0}.
\end{align}

By equating the primary and dual local Boltzmann factors, enforcing the principal Boltzmann condition ($l=0$) and taking the replica limit $N_{\mathrm{rep}}\to 0$, and by retaining the leading term in $N_{\mathrm{rep}}$, we obtain the implicit curve for the entire phase boundary
\begin{align}
    (1-2p)(1-q)K=&(1-q)\left[\log(2\cosh(K))-\frac{1}{2}\log(2)\right]\nonumber\\
    &+ q\frac{1}{2}\log(2).
\end{align}

If we set $q=0$, we recover the phase boundary predicted in Eq.~\eqref{eq:phase_boundary}. If we set $p=0$, then in taking the limit $K\to\infty$, i.e., at zero temperature, we recover $q_{th}=1/2$, consistent with the value set by the no-cloning theorem \cite{PhysRevLett.78.3217}. We plot in Fig.~\ref{fig:3d_phase_diagram} the predicted phase diagram incorporating bit-flips and qubit erasures using the replica trick.

\begin{figure}[t]
    \centering
    \includegraphics[width=\linewidth]{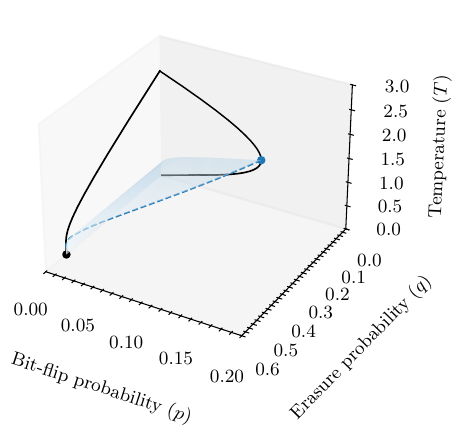}
    \caption{The predicted phase diagram incorporating bit-flips and qubit erasures using the replica trick. The intersection of the Nishimori surface, defined by $K=J\beta=\frac{1}{2}\log\left(\frac{1-p}{p}\right)$, with the phase diagram is indicated by the dashed blue line. The blue circle indicates the multicritical Nishimori point, and the black circle indicates the zero-temperature, bond-dilution critical point at $p=0,T=0,q=1/2$.}
    \label{fig:3d_phase_diagram}
\end{figure}

We project the intersection of the Nishimori surface, which is parametrized by $K_{pq}=\beta J=\frac{1}{2}\log\left(\frac{1-p}{p}\right)$ onto the $p,q$ plane, which represents the optimal threshold under mixtures of both bit-flips and qubit erasures. Here, $q$ denotes the erasure probability, while $p$ denotes the bit-flip probability conditioned on the qubit not being erased. With this convention the Nishimori conditions fix only the ratio of antiferromagnetic to ferromagnetic bonds. This differs from the convention of Ref.~\cite{Picco_2006}, where $p'$ denotes the unconditional antiferromagnetic bond probability; the two parameterizations are related by $p'=(1-q)p$. Our convention matches the mixed bit-flip and erasure setting of Ref.~\cite{PhysRevLett.102.200501}. This is plotted in Fig.~\ref{fig:optimal_bitflip_erasure}. Although Ref.~\cite{PhysRevLett.102.200501} considers the zero-temperature phase diagram of a surface code obtained using minimum-weight perfect matching decoding, the resulting boundary is strikingly similar to the present optimal decoding prediction: in both cases the loss threshold approaches $q=0.5$, and the phase boundary exhibits the same overall shape, including an apparent inflection point (at around $q\approx 0.4$).

\begin{figure}[t]
    \centering
    \includegraphics[width=\linewidth]{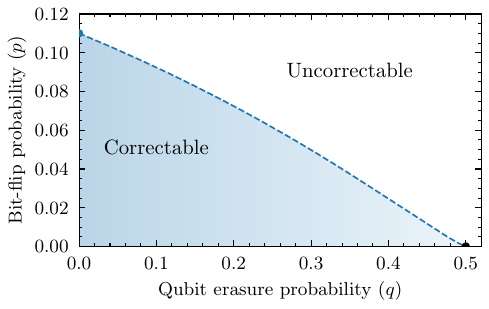}
    \caption{Correctability phase diagram obtained by the intersection of the Nishimori surface with the phase boundary predicted by the replica trick. The shaded region is correctable in the thermodynamic limit.}
    \label{fig:optimal_bitflip_erasure}
\end{figure}

\section{Concatenated quantum codes, renormalization group flow, and finite-size crossover}\label{sec:concat}

In Section~\ref{sec:symmetry} we showed that for the zero-rate \emph{em}-symmetric families considered here, the optimal code capacity threshold is strongly constrained by generalized Kramers-Wannier self-duality. This shifts the comparison between \emph{em}-symmetric topological and concatenated code families away from the threshold location itself and toward their finite-size behaviour below threshold. Concatenated codes are particularly convenient to analyse in this setting because under the statistical mechanical mapping of optimal decoding \cite{Dennis_02,Chubb_2021}, they generate classical models on hierarchical lattices. In the fully postselected limit, where the quenched disorder is removed, this hierarchical structure yields an exact real-space renormalization group (RG) flow \cite{PhysRevB.26.5022}.

In this section, we investigate the finite size behaviour of zero-rate \emph{em}-symmetric codes. We first show that optimal decoding of concatenated codes can be written as a hierarchical message-passing procedure, which reduces to an exact recursion on logical sector partition functions under full postselection. We then show that one of the main distinctions between topological and concatenated \emph{em}-symmetric codes is a tradeoff between distance scaling and the entropy of minimum-weight logical operators. Finally, we show how this tradeoff is manifest at finite size through crossover maps in physical overhead, including a band structure induced by the discrete coarse grained distances of concatenated families.

\subsection{Optimal decoding of concatenated codes as a renormalization group flow}\label{subsec:algorithm}

Optimal decoding of concatenated quantum error correcting codes was shown to be possible with a message passing algorithm that is linear in the number of physical qubits in Ref.~\cite{PhysRevA.74.052333}. For a code encoding a single logical qubit with bit-flip noise, optimal decoding in the code capacity setting reduces to comparing the posterior weights of the two logical cosets. In the statistical mechanical mapping this is equivalent to comparing the two sector partition functions $Z_{\vec{e},\vec{0}}(\theta;K_{p})$ and $Z_{\vec{e\mathcal{L}_{X}},\vec{0}}(\theta;K_{p})$ for a representative disorder instance $\vec{e}$ determined by the measured syndrome. When the code admits a transversal Pauli representative supported on all physical qubits, the two sectors obey the global symmetry
\begin{equation}
    Z_{\vec{e\mathcal{L}_{X}},\vec{0}}(\theta;K_{p})=Z_{\vec{e},\vec{0}}(\theta;-K_{p}),
\end{equation}
so it suffices to evaluate the partition function at equal magnitude positive and negative temperature: $Z_{\vec{e},\vec{0}}(\theta;K_{p})$ and $Z_{\vec{e},\vec{0}}(\theta;-K_{p})$.

We consider a seed code that maps to a generalized Ising model as in Eq.~\eqref{eq:GIM} with $N_s$ Ising spins with $S_{i}=\pm1$ and $N_{b}=n$ interaction terms (one per physical qubit of the seed code). For a single layer of concatenation, the Hilbert space of each physical qubit is replaced with the logical subspace of a copy of the base code. In doing so, the statistical mechanical model multiplies each term by a copy of the base Hamiltonian. Under concatenation, assuming a transversal Pauli logical operator, the Hamiltonian transforms as
\begin{equation}
    \mathcal{H}\to-\sum_{b=1}^{n}J\left(\sum_{b'=1}^{n}\prod_{j'}S_{j'}^{\theta_{jb}}\right)\left[\prod_{j}S_{j}^{\theta_{jb}}\right].
\end{equation}
This concatenation can also be captured by the transformation of $\theta$ expressed in Eq.~\eqref{eq:theta_concat}. We use this expression to motivate the decoding algorithm below.

We now consider $r$ levels of concatenation, where each term of a parent block is multiplied by an independent child block described by the same seed Hamiltonian. Conditioned on the parent spin configuration, the child blocks are independent. As a result, each child block can be summarized by a two component message
\begin{equation}
    u \defeq (Z_{\vec{e},\vec{0}}(\theta;K_{p}),Z_{\vec{e},\vec{0}}(\theta;-K_{p})),
\end{equation}
and parent messages are computed by combining the $n$ child messages through a constant size sum over the $2^{N_s}$ seed spin configurations. This is equivalent to the message-passing decoder of Ref.~\cite{PhysRevA.74.052333} expressed in the statistical mechanical language. Here, we outline a generic algorithm for optimal decoding of a single logical qubit quantum error correcting code in the bit-flip code capacity setting, assuming a transversal Pauli representative:

\begin{figure}
\begin{mdframed}[
  frametitle={Hierarchical maximum-likelihood decoder},
  roundcorner=2pt,
  innertopmargin=0.7\baselineskip,
  innerbottommargin=0.7\baselineskip,
  splittopskip=\baselineskip,
  splitbottomskip=\baselineskip,
]
\begin{itemize}[leftmargin=*,itemsep=0.25em]
\item \textbf{Inputs:} seed monomials $\{\prod_{j}S_{j}^{\theta_{jb}}\}_{b=1}^n$;
number of concatenation levels $r$; representative error $\vec{e}$ at the leaf couplings.
\item \textbf{Initialization ($\ell=1$):} for each of the $n^{r-1}$ leaf blocks, indexed by $l$, with disorder $\vec{e}$ compute
\begin{align}
    Z^{(1,l)}_{\vec{e},\vec{0}}(\theta;K_{p}) &= \sum_{S_{i}=\pm 1}\prod_{b=1}^{n}
    \exp\left(K_{p}(-1)^{e_{b}}\prod_{j}S_{j}^{\theta_{jb}}\right),\\
    Z^{(1,l)}_{\vec{e},\vec{0}}(\theta;-K_{p}) &= \sum_{S_{i}=\pm 1}\prod_{b=1}^{n}
    \exp\left(-K_{p}(-1)^{e_{b}}\prod_{j}S_{j}^{\theta_{jb}}\right).
\end{align}
Store the messages\\ $u^{(1,l)} = (Z_{\vec{e},\vec{0}}(\theta;K_{p}),Z_{\vec{e},\vec{0}}(\theta;-K_{p}))$.
\item \textbf{For each $\ell=2,\dots,r$:}
group the messages from $n$ children into each parent block. For a parent block with children, e.g., $u^{(1,1)},\dots,u^{(1,n)}$, compute its identity sector partition function
\begin{equation}
    Z_{\vec{e},\vec{0}}^{(\ell,l)}(\theta;K_{p})=\sum_{S_{i}=\pm1}
    \prod_{b=1}^n Z^{(\ell-1,l)}_{\vec{e},\vec{0}}(\theta;s(j,\vec{S})),
\end{equation}
where $s(j,\vec{S})=\pm K_{p}$ is selected by the seed monomial sign:
\begin{equation}
    s(j,\vec{S})=
    \begin{cases}
    K_{p},& \prod_{j}S_{j}^{\theta_{jb}}=+1,\\
    -K_{p},& \prod_{j}S_{j}^{\theta_{jb}}=-1.
    \end{cases}
\end{equation}
The $X$-sector partition function is obtained by swapping $K_{p}\leftrightarrow -K_{p}$ in the selector
(or equivalently, by using the transversal symmetry for the whole parent block):
\begin{equation}
    Z_{\vec{e},\vec{0}}^{(\ell,l)}(\theta;-K_{p})=\sum_{S_{i}=\pm 1}
    \prod_{b=1}^n Z^{(\ell-1,l)}_{\vec{e},\vec{0}}(\theta;\bar s(j,\vec{S})),
\end{equation}
with $\bar s$ the opposite sector. Output the parent message $u^{(\ell,l)}=(Z^{(\ell,l)}_{\vec{e},\vec{0}}(\theta;K_{p}),Z^{(\ell,l)}_{\vec{e},\vec{0}}(\theta;-K_{p}))$.
\item \textbf{Outputs:}
after $r$ levels, a single root message \\$u^{(r)}=(Z^{(r)}_{\vec{e},\vec{0}}(\theta;K_{p}),Z^{(r)}_{\vec{e},\vec{0}}(\theta;-K_{p}))$ remains. Correct into the logical coset represented by the largest partition function.
\end{itemize}
\end{mdframed}
\end{figure}

Let $N_{b}$ denote the total number of interaction terms (couplings) in the full concatenated Hamiltonian. For an $r$-level concatenation of a seed with $n$ couplings, $N_{b}=n^r$. Each block update enumerates $2^{N_s}$ seed spin configurations (with $N_{s}$ spins in the seed code) and performs $\mathcal{O}(n)$ multiplications, hence costs $\mathcal{O}(n 2^{N_s})$, which is a constant for a fixed seed code\footnote{This is not constant for concatenated quantum code constructions with seed codes whose number of physical qubits grows with the concatenation layer, such as the Hamming code construction in Refs.~\cite{Yamasaki2024,Yoshida2025}.}. The number of blocks across all levels is a geometric series
\begin{equation}
    \#\text{blocks}=\sum_{i=0}^{r-1} n^{i}=\frac{n^{r}-1}{n-1}=\mathcal{O}(N_{b}),
\end{equation}
so the total runtime per disorder instance (i.e., per fixed physical error sample) is
\begin{equation}
    T_{\mathrm{decode}}=\mathcal{O}(N_{b}),
\end{equation}
and the memory can be straightforwardly made $\mathcal{O}(N_{b})$. When estimating quenched averages by direct sampling over $N_{\mathrm{samp}}$ instances, the total runtime scales as $\mathcal{O}(N_{\mathrm{samp}} N_{b})$.

In the fully postselected limit, the disorder is removed and all couplings in the identity sector are ferromagnetic. In this case the recursion closes exactly on the logical sector partition functions, so the decoding problem becomes an analytically solvable real-space RG flow. We derive the corresponding recursion relations in Appendix~\ref{appendix:RG_flow}. This exact solvability makes concatenated codes a particularly clean setting in which to isolate the structural origin of the finite-size crossover discussed below.

\subsection{Geometry and entropy of topological and concatenated growth}\label{sec:comparison}

If the thresholds of the relevant code families are already strongly constrained by duality, what then distinguishes their sub-threshold behaviour? In the low error rate limit, the distance determines the sub-threshold scaling. However, at higher error rates in the sub-threshold regime one must also account for the multiplicity of the lowest-weight logical operators. We therefore compare both the distance and the entropy of minimum-weight logical operators for topological and concatenated growth.

Consider a family of stabilizer codes indexed by the number of physical qubits $n$. Let $d(n)$ denote the code distance and define the relative distance
\begin{equation}
    \delta(n)\defeq\frac{d(n)}{n}.
\end{equation}
We also introduce an entropy measure for the lowest-weight logical operators. Let $N_{\min}(n)$ be the number of distinct minimum-weight logical Pauli operators, and define
\begin{equation}
    S_{\min}(n)\defeq\log N_{\min}(n).
\end{equation}
We now compare the scaling of these quantities for topological and concatenated code families.
For a $D$-dimensional local stabilizer code on a lattice of linear size $L$, one generically has $n=\Theta(L^{D})$. Locality further imposes the upper bound $d=\mathcal{O}(L^{D-1})$ \cite{Bravyi_2009}, equivalently $d(n)=\mathcal{O}(n^{\frac{D-1}{D}})$. In $D=1,2,3,4,5$ this bound is tight~\cite{Williamson2024,yuan2026}. For the standard $2D$ toric/surface code families, we have
\begin{equation}
    n=\Theta(L^{2}),\qquad d=\Theta(L),
\end{equation}
so that
\begin{equation}
    d(n)=\Theta(\sqrt{n}),\qquad\delta(n)=\Theta(n^{-1/2}).
\end{equation}

A common feature of 2D topological codes is that the number of minimum-length nontrivial logical strings is at most polynomial in $L$, and hence polynomial in $n$. For example, in an $L\times L$ toric code the shortest non-contractible loops may be translated across the lattice, yielding only $\mathcal{O}(L)$ distinct translates in each homology class. Thus one expects, schematically,
\begin{align}\label{eq:entropy_topological}
    N_{\min}(n)&\sim \poly(L)\sim \poly(\sqrt n),\nonumber\\
    S_{\min}(n)&=\log N_{\min}(n)\sim \mathcal{O}(\log n),
\end{align}
that is, only subextensive entropy for minimum-weight logical operators.

Now, consider a concatenated family built from a fixed seed code $[[n_0,1,d_0]]$. After $r$ levels of concatenation,
\begin{equation}
    n=n_0^r,\qquad d=d_0^r.
\end{equation}
Eliminating $r=\log_{n_0}(n)$ gives
\begin{align}
    d(n)&=n^{\log_{n_0}(d_0)},\nonumber\\
    \delta(n)&=n^{\log_{n_0}(d_0)-1}\xrightarrow[n\to\infty]{}0,
\end{align}
since $d_0<n_0$ implies $\log_{n_0}(d_0)<1$. For the concatenated $[[9,1,3]]$ surface-17 code seed, one obtains $d(n)=n^{\log_9 3}$ and $\delta(n)\sim n^{-1/2}$, matching the distance scaling of a 2D topological code. For other seeds, such as Steane's $[[7,1,3]]$ code, one obtains the improved scaling \(d(n)=n^{\log_7 3}\approx n^{0.565}\).

If distance scaling were the only factor determining performance, concatenated codes would therefore appear at least competitive with topological codes, and in some cases asymptotically superior. Let us next examine how the entropy of minimum-weight logical operators can change this picture.

Let $L_0$ be the number of minimum-weight logical Pauli operators of the seed. A minimum-weight logical operator at level $r$ is obtained by choosing a minimum-weight outer logical and then, for each of the $d_0$ affected blocks, choosing a minimum-weight logical operator of the level-$(r-1)$ code. This yields the recursion relation 
\begin{equation}
    N_{\min}^{(r)} = L_0\left(N_{\min}^{(r-1)}\right)^{d_0},
    \qquad N_{\min}^{(1)}=L_0,
\end{equation}
whose closed-form solution is
\begin{align}
    N_{\min}^{(r)}
    &=L_0^{\frac{d_0^r-1}{d_0-1}}
    =\exp\left(\frac{\log L_0}{d_0-1}d_0^r+\mathcal{O}(1)\right)\nonumber\\
    &=\exp\left(s_0 d(n)+\mathcal{O}(1)\right),
    \qquad
    s_0\defeq\frac{\log L_0}{d_0-1}.
\end{align}
Equivalently,
\begin{align}
    N_{\min}(n)&\sim \exp\left(s_0 n^{\log_{n_0}(d_0)}\right),\nonumber\\
    S_{\min}(n)&\sim s_0 n^{\log_{n_0}(d_0)}.
\end{align}
This uncovers the key distinction between concatenated and topological codes. In 2D topological codes, minimum-weight logical operators are expected to have subextensive entropy, $S_{\min}(n)\sim \mathcal{O}(\log n)$. In concatenated codes, the same quantity grows linearly on the distance scale, $S_{\min}(n)\sim \mathcal{O}(d(n))$. Thus concatenated growth produces exponentially many minimum-weight failure mechanisms on the code-distance scale.

In summary, both topological and concatenated code families satisfy $\delta(n)\to0$, but with very different scaling of the low-weight logical entropy.
\begin{align}\label{eq:parameter_scaling}
    \text{2D topological:}\quad & d(n)\sim n^{1/2},\nonumber\\
    &\delta(n)\sim n^{-1/2},\nonumber\\
    &S_{\min}(n)\sim \mathcal{O}(\log(n)),\nonumber\\
    \text{Concatenated:}\quad & d(n)\sim n^{\log_{n_0}(d_0)},\nonumber\\
    &\delta(n)\sim n^{\log_{n_0}(d_0)-1},\nonumber\\
    &S_{\min}(n)\sim \mathcal{O}(d(n)).
\end{align}

In the path-counting regime \cite{Dennis_02,Watson_2014,Fowler_12,Beverland_2019}, this distinction enters directly into the leading estimate of the logical failure probability:
\begin{equation}\label{eq:path_counting}
    \mathbb{P}_{\mathrm{fail}}(p,d)\approx N_{\min}(d)\binom{d}{\lceil d/2\rceil} p^{\lceil d/2\rceil}.
\end{equation}
The much larger growth of $N_{\min}(d)$ in concatenated families therefore produces an entropic penalty. This means that at fixed distance, concatenated families can support many more low-weight failure mechanisms than Euclidean topological ones, even when their distance scaling with $n$ is favorable. 
This suggests the following tradeoff. At fixed distance $d$, topological codes are expected to exhibit a larger sub-threshold decay rate under optimal decoding. At fixed physical overhead $n$, however, a concatenated family with sufficiently favorable distance scaling can eventually compensate for this entropic penalty. This competition naturally leads to a crossover scale in physical overhead.

We now estimate this crossover scale in the path-counting limit. Using Stirling's approximation for the central binomial coefficient,
\(
\binom{d}{\frac{d}{2}}=\Theta(2^{d+1}/\sqrt{2\pi d})
\),
Eq.~\eqref{eq:path_counting} becomes, up to polynomial factors,
\begin{equation}\label{eq:stirling_simple}
    \mathbb{P}_{\mathrm{fail}}(n,p)\approx \frac{N_{\min}(n)}{\poly(d(n))}\,(2\sqrt p)^{d(n)}.
\end{equation}

For a planar surface code encoding one logical qubit, $N_{\min}^{\mathrm{surf}}(n)=\poly(d_{\mathrm{surf}}(n))$, so
\begin{equation}\label{eq:pfail_surf_simple}
    \log \mathbb{P}_{\mathrm{fail}}^{\mathrm{surf}}(n,p)=\sqrt n\,\log(2\sqrt p)+\mathcal{O}(\log n).
\end{equation}

For the concatenated Steane family,
\begin{align}
    &n=7^r,\qquad d=3^r,\qquad d_{\mathrm{con}}(n)=n^\kappa,\nonumber\\
    &\kappa\defeq\log_7 3\approx0.565.
\end{align}
Moreover, $N_{\min}^{\mathrm{con}}(n)\sim \exp(s_0 d_{\mathrm{con}}(n))$, so
\begin{equation}\label{eq:pfail_con_simple}
    \log \mathbb{P}_{\mathrm{fail}}^{\mathrm{con}}(n,p)=n^\kappa\big(\log(2\sqrt p)+s_0\big)+\mathcal{O}(\log n).
\end{equation}

Now we define the crossover length $n^*(p)$ by equating the leading terms in Eqs.~\eqref{eq:pfail_surf_simple} and \eqref{eq:pfail_con_simple}. This is the number of physical qubits at fixed $p$ above which the failure rate of the concatenated Steane code is expected to be lower than the planar surface code in the path-counting regime. This crossover length is given by
\begin{equation}\label{eq:nstar_equation_simple}
    \sqrt{n^*}\,|\log(2\sqrt p)|\approx (n^*)^\kappa |\log(2\sqrt p)+s_0|.
\end{equation}
This yields
\begin{equation}\label{eq:nstar_closed_simple}
    n^*(p)\approx\left[\frac{|\log(2\sqrt p)|}{|\log(2\sqrt p)+s_0|}\right]^{\frac{1}{\kappa-\frac12}}.
\end{equation}
Since $\kappa-\tfrac12\approx0.065$, the exponent $(\kappa-\tfrac12)^{-1}\approx15.5$ is large, so the crossover is extremely sensitive to subleading constants, especially the entropic penalty $s_0$, even though $d_{\mathrm{con}}(n)$ grows asymptotically faster than $d_{\mathrm{surf}}(n)$. We show $n^*(p)$ in Fig.~\ref{fig:nstar_plot}. As $p\to0$, the ratio inside Eq.~\eqref{eq:nstar_closed_simple} approaches $1$, and hence $n^*(p)\to1$.

\begin{figure}[t]
    \centering
    \includegraphics[width=\linewidth]{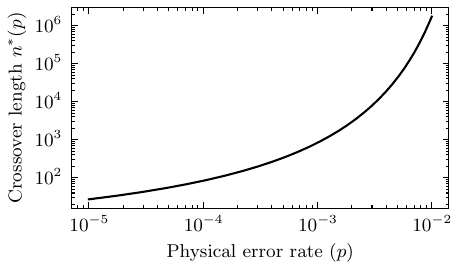}
    \caption{Estimated crossover length \(n^*(p)\) for the concatenated Steane code relative to a planar surface code in the path-counting regime. For \(n\gtrsim n^*(p)\), the improved distance scaling of the concatenated family is expected to overcome its entropic penalty.}
    \label{fig:nstar_plot}
\end{figure}

\subsection{Finite-size crossover and band structure}\label{subsec:crossover}

We now turn from the asymptotic scaling picture to the actual finite-size crossover in physical overhead. In this section, we only examine the fully postselected limit of the codes considered to enable exact analytic results. We compare the physical overhead of concatenated Steane codes and toric codes. Here, we consider postselected codes, as numerical estimates of the failure rates required for non-postselected codes do not readily allow for accurate extrapolation to arbitrary code distances, as noted by the system size dependence of the decay rate observed in Ref.~\cite{Beverland_2019}.

For a target logical failure rate $\epsilon$ and physical error rate $p$, let
\begin{equation}
    n_{\mathrm{fam}}(p,\epsilon)\defeq \min\{n:\mathbb{P}_{\mathrm{fail}}^{\mathrm{fam}}(n,p)\le \epsilon\}
\end{equation}
denote the smallest block size in a given code family achieving the target. We compare code families through the quantity
\begin{equation}
    \Delta(p,\epsilon)\defeq \log\left(\frac{n_{\mathrm{steane}}(p,\epsilon)}{n_{\mathrm{toric}}(p,\epsilon)}\right),
\end{equation}
so that $\Delta<0$ indicates that the concatenated Steane code uses fewer physical qubits, while $\Delta>0$ indicates that the toric code is more efficient.

\begin{figure}[t]
    \centering
    \includegraphics[width=\linewidth]{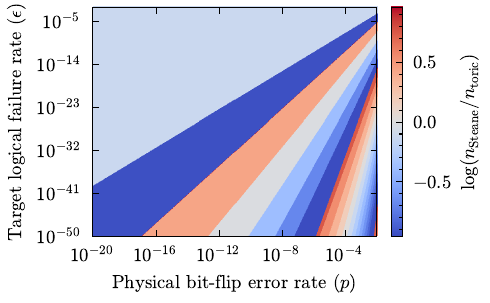}
    \caption{Code family crossover in the fully postselected setting. Blue indicates the concatenated Steane code requires fewer physical qubits than the toric code to achieve a target logical failure rate; red indicates the converse.}
    \label{fig:postselect_crossover}
\end{figure}

Fig.~\ref{fig:postselect_crossover} shows that the competition between the two code families is not described by a single monotone boundary at finite size, but instead by a banded structure. The bands arise due to the large discrete steps in system size that occur in concatenated code growth. A concatenated family only realizes the sequence of block sizes $n=n_0^r$ and distances $d=d_0^r$, so the attainable logical failure rates change in a coarse grained manner as the concatenation level increases. By contrast, topological growth proceeds through a denser sequence of distances. As the target logical failure rate is varied, the optimal concatenation level changes discontinuously, producing alternating regions in which one code family or the other minimizes the physical overhead.

This interpretation is consistent with the scaling analysis of the previous subsection. The asymptotic estimate $n^*(p)$ describes the crossover number of physical qubits in the path-counting regime, while the individual bands arise from the discreteness of the different code distances. Moreover, because Eq.~\eqref{eq:nstar_closed_simple} implies $n^*(p)\to1$ as $p\to0$, the multiple bands are expected to collapse toward a single asymptotic band in the deep path-counting regime. In this sense, the low-$p$ limit is controlled by a single crossover scale, while the richer band structure seen at finite $p$ reflects finite-size effects specific to topological and concatenated codes.

Taken together, these results show that once the threshold is constrained by duality, the relevant distinction between topological and concatenated code families is a finite-size tradeoff of code distance and multiplicity of logical operators. Topological codes benefit from a much smaller entropy of minimum-weight logical operators, while concatenated codes can eventually leverage superior distance scaling with $n$. The resulting competition is thus naturally expressed as a crossover in physical overhead.

\section{Thresholds of quantum two-block group algebra codes}\label{sec:BB_codes}

In this section, we discuss an extension of the duality-constrained threshold phenomenology developed above beyond topological and concatenated codes to more general qLDPC code families. The preceding arguments do not rely on geometric locality, nor on any structure specific to concatenated codes. Rather, they require only three key ingredients: \emph{em}-symmetry of the CSS codes, asymptotically vanishing rate, and the existence of a unique decoding threshold. The structure of abelian two-block group algebra (A2BGA) codes~\cite{PhysRevA.109.022407}, which include the bivariate bicycle (BB) codes~\cite{Bravyi2024}, makes them a natural class of qLDPC codes to which we can apply these arguments. We conclude the section by examining a family of BB codes numerically under code capacity bit-flip noise.

For A2BGA codes, the CSS check matrices take the block form
\begin{equation}
    H_X = [A\ |\ B],
    \qquad
    H_Z = [B^{T}\ |\ A^{T}],
\end{equation}
where \(A\) and \(B\) are elements of an abelian group algebra. Up to exchanging the two block components, \(H_Z\) is therefore related to \(H_X\) by row and column permutations, and hence \emph{em}-symmetric. Thus A2BGA codes are natural candidates for applying the results of Theorem~\ref{thm:1}.

Next, we invoke the higher-dimensional parent code picture of Refs.~\cite{hopkin2026,sym14071348,shaw2026}, where translation-invariant qLDPC codes with fixed check weight were shown to arise as compactifications of local fracton codes in higher dimensions. In particular, A2BGA codes were shown to be compactifications of hypergraph product (HGP) fracton codes~\cite{f48m-rlh3}. 
This implies that all BB codes with a given fixed check weight can be viewed as descendants of a single parent model which, on an infinite lattice, one may consider as the thermodynamic limit for the whole family.

The compactified parent code construction implies that A2BGA code families with constant check weight $w$ obey the code parameter tradeoff bound for local codes in a higher dimension~\cite{PhysRevLett.104.050503}
\begin{align}
    k d^{\frac{2}{D-1}} \leq O(n),
\end{align}
where $D=w-2$. 
Due to this tradeoff, any A2BGA code with constant check weight and polynomially growing distance has asymptotically vanishing rate. 
Hence, our arguments apply to any A2BGA code family with constant check weight and polynomially growing distance, provided there is a unique error correction threshold. For such zero-rate \emph{em}-symmetric families, the same logic as in Corollary~\ref{thm:2} implies that sector-mixing defects are subextensive, so that the clean model is thermodynamically self-dual and the critical point is pinned to the self-dual temperature. In this case, the corresponding disordered phase boundary above the Nishimori line is again constrained by the principal Boltzmann factor construction.

As a special case, any BB code specified by fixed finite generating polynomials can be placed on larger and larger tori by increasing the cycle lengths while keeping the polynomial supports fixed. In this limit the stabilizer interaction pattern remains finite range in the resulting translation-invariant lattice model. Consequently, such a thermodynamic limit should be understood as a geometrically local parent model rather than as a genuinely non-local qLDPC code family. In particular, fixed-polynomial compactifications remain in the same broad topological code setting discussed in Refs.~\cite{PhysRevB.99.245135,rmy6-9n89,86j7-cmsw}.

To investigate A2BGA codes beyond this topological setting, we instead consider sequences in which the polynomial supports themselves grow with the code size. This is necessary to achieve performance beyond that of a 2D topological code as the family is scaled up. We let
\[
    A_n = 1+f_n(\mathbf{x}),\qquad B_n = 1+g_n(\mathbf{y}),
\]
where \(f_n\) and \(g_n\) are finite sums of monomials whose degrees increase along the sequence. Using the notation of Yoshida~\cite{PhysRevB.88.125122}, the CSS stabilizer checks of the corresponding BB code family may be written as
\begin{equation}
    X\begin{pmatrix}
        1+f_n(\mathbf{x})\\
        1+g_n(\mathbf{y})
    \end{pmatrix},
    \qquad
    Z\begin{pmatrix}
        \overline{1+g_n(\mathbf{y})}\\
        \overline{1+f_n(\mathbf{x})}
    \end{pmatrix},
\end{equation}
where the bar denotes the group algebra involution, sending each translation variable to its inverse. The condition
\[
    \deg(f_n),\deg(g_n) \to \infty
\]
should be interpreted as increasing the range of the stabilizer pattern in the chosen compactified presentation. This distinguishes the sequence from a fixed local parent model on a growing torus.

The relevance of this distinction is that the duality argument developed above does not require geometric locality. Rather, it requires only \emph{em}-symmetry, vanishing asymptotic rate, and a unique thermodynamic transition. A2BGA codes, and their BB code subclass, satisfy the first condition by construction. A2BGA codes with fixed check weight and polynomially growing distance also satisfy the second condition. Hence, these duality arguments apply to any A2BGA code with fixed check weight, a polynomially growing distance, and a unique threshold. This is expected to capture all A2BGA code families that are being considered as candidates for qLDPC architectures. 

\subsection{Fully postselected A2BGA codes}

We first study the fully postselected limit of A2BGA codes. Full postselection removes the quenched disorder by conditioning on the trivial syndrome, so the decoding problem reduces to comparing logical sector partition functions of the clean statistical mechanical model. This setting is useful because the duality prediction is exact: if the clean model is thermodynamically self-dual and has a unique transition, then the transition occurs at the self-dual coupling given in Eq.~\eqref{eq:K_c}. Under the Nishimori conditions, this corresponds to
\[
    p_c=\frac{1}{2+\sqrt{2}}\approx 0.2929.
\]

Because the BB code instances considered below encode different numbers of logical qubits, the block logical failure rate \(\mathbb{P}_{\mathrm{fail}}\) is not the most convenient quantity for comparing finite sizes. We therefore use the per-logical proxy
\begin{equation}\label{eq:per_logical_proxy}
    y_n(p)\defeq 1-(1-\mathbb{P}_{\mathrm{fail}})^{1/k}.
\end{equation}
This is the effective identical per-logical failure probability that would reproduce the same total block success probability if the \(k\) logical degrees of freedom failed independently and symmetrically. When the logical failure rate is small,
\[
    y_n(p)\approx \frac{\mathbb{P}_{\mathrm{fail}}}{k},
\]
so this normalization removes the leading trivial dependence on the number of encoded logical qubits. In the high-temperature limit, where all \(2^k\) logical sectors become asymptotically equiprobable, this proxy tends to \(1/2\), as expected for random guessing of each logical degree of freedom.

We characterize the non-locality of the BB presentation by the maximum degree of the polynomials. For the polynomials
\[
    A(x,y)=\sum_{(a,b)\in S_A} x^a y^b,
    \qquad
    B(x,y)=\sum_{(a,b)\in S_B} x^a y^b,
\]
where $S_{A} (S_{B})$ denote the set of terms in the A(B) blocks, we define the range
\begin{equation}
    \rho \defeq \max(\deg(A),\deg(B)),
\end{equation}
using the non-negative exponent representatives appearing in the chosen presentation. Equivalently, the range \(\rho\) is the largest total monomial degree appearing in either generator. For a code defined on an $\ell\times m$ torus, this range may not be the minimal diameter of the checks since exponents are only defined modulo the relations \(x^\ell=y^m=1\). Rather, it is a convenient presentation-dependent upper bound on the minimum range required to implement the code that is expected to become tight when the degree of the interaction polynomials is sufficiently smaller than the length of the periodic boundary conditions. 
An increasing range indicates that a given sequence of codes is not geometrically local in a fixed two-dimensional embedding.

We investigate a family of weight-6 BB codes, and obtain the fully postselected logical failure rates per-logical by estimating the partition functions in each logical sector using a Markov chain Monte Carlo (MCMC) with population annealing, see Appendix~\ref{app:MCMC_pop_annealing} for details. The code parameters for the codes studied are listed in Tab.~\ref{tab:BB6_parameters}. These examples were chosen to include increasing polynomial range while keeping the stabilizer weight fixed.

\begin{table}[]
    \centering
    \begin{tabular}{c|c|c|c}
    \hline
    $[[n,k,d]]$ & $A$ & $B$ & $\rho$ \\
    \hline
    $[[18,4,4]]$ & $1+x+y$ & $1+x^2+y^2$ & $2$ \\
    $[[36,4,6]]$ & $x+y^2+y^3$ & $1+y+x^2$ & $3$ \\
    $[[54,4,8]]$ & $x+y+y^3$ & $1+y^2+x^2$ & $3$ \\
    $[[28,6,4]]$ & $1+xy+xy^3$ & $1+xy+y^3$ & $4$ \\
    $[[30,4,6]]$ & $1+xy+(xy)^2$ & $1+(xy)^2+xy^2$ & $4$
    \end{tabular}
    \caption{Bivariate bicycle codes used in the population annealing simulations. The polynomial interaction range $\rho$ is the largest total monomial degree appearing in $A$ or $B$.}
    \label{tab:BB6_parameters}
\end{table}

We show in Fig.~\ref{fig:BB_codes_clean} the per-logical proxy \(y_n(p)\), see Eq.~\eqref{eq:per_logical_proxy}, for this family under full postselection. For each value of \(p\), the clean inverse temperature is set by the Nishimori conditions, see Eq.~\eqref{eq:nishimori_conditions}. The logical sector partition functions are estimated separately using population annealing MCMC and the block failure probability is computed from the normalized weight outside the identity logical sector,
\[
    \mathbb{P}_{\mathrm{fail}}(p)
    =
    1-\frac{Z_{\vec{0}}(\theta;K_{p})}{\sum_{\lambda\in\mathbb{F}_2^k}Z_{\lambda}(\theta;K_{p})}.
\]

\begin{figure}[t]
    \centering
    \includegraphics[width=\linewidth]{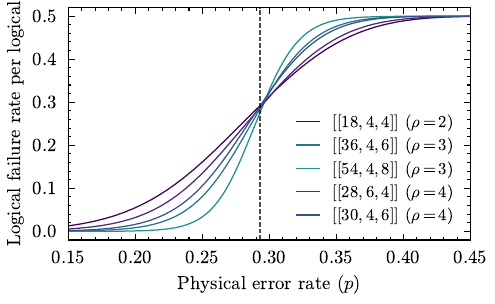}
    \caption{Per-logical proxy failure rate \(y_n(p)=1-(1-\mathbb{P}_{\mathrm{fail}})^{1/k}\) for a family of weight-6 bivariate bicycle codes under full postselection with code capacity bit-flip noise. The vertical dashed line marks the clean self-dual prediction \(p_c=1/(2+\sqrt{2})\approx0.2929\). The approximate finite-size crossing near this value is consistent with thermodynamic self-duality of zero-rate \emph{em}-symmetric CSS families.}
    \label{fig:BB_codes_clean}
\end{figure}

The curves cross close to the predicted value $p_c\approx0.2929$, consistent with the clean self-duality prediction. The residual spread of the crossings should be interpreted as a finite-size and sampling effect rather than as evidence for distinct asymptotic critical points. In particular, these are small non-geometrically local instances, and the population annealing estimates are most delicate near the transition.

We next consider Haah's cubic code \cite{PhysRevA.83.042330}, a prototypical type-II fracton code that is also an A2BGA code. Its associated clean statistical mechanical model is known as the fractal Ising model~\cite{PhysRevB.94.155128,PhysRevB.94.235157}, and is self-dual, possessing a unique critical self-dual coupling as in Eq.~\eqref{eq:K_c} \cite{PhysRevResearch.6.013304}. This model is numerically challenging because of its fractal constraints and slow equilibration near the transition \cite{canossa2025}. We therefore restrict to small lattice sizes with logical qubits $k\leq 6$ \cite{PhysRevA.83.042330}. Fig.~\ref{fig:Haah_code_clean} shows the fully postselected logical failure rates for the smallest available instances in this subsequence. The data are again obtained using population annealing, with numerical details given in Appendix~\ref{app:MCMC_pop_annealing}. Although the accessible system sizes are small, the approach to a step function centered on the predicted self-dual value provides an additional test of the same duality-constrained clean critical point.

\begin{figure}[t]
    \centering
    \includegraphics[width=\linewidth]{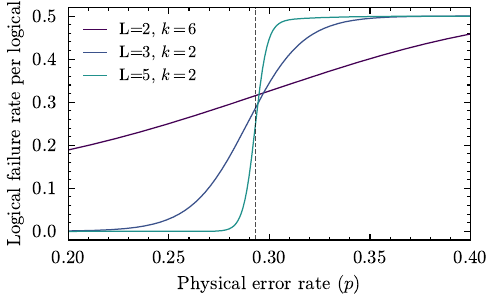}
    \caption{Per-logical proxy failure rate \(y_n(p)=1-(1-\mathbb{P}_{\mathrm{fail}})^{1/k}\) for Haah's cubic code under full postselection with code capacity bit-flip noise. The vertical dashed line marks the clean self-dual prediction \(p_c=1/(2+\sqrt{2})\approx0.2929\). The approach to a step function is consistent with the predicted self-dual critical point.}
    \label{fig:Haah_code_clean}
\end{figure}

With clean critical points constrained, and similarly the optimal non-postselected thresholds constrained at $p_{c}\approx 0.11$ by the principal Boltzmann factor construction, we next investigate whether similar constraints are placed on practical sub-optimal decoders without postselection. 

\subsection{Non-postselected A2BGA codes}

We now turn to the non-postselected setting, and consider statistical mechanics models where a quenched average is taken over the full disorder landscape. We investigate practical decoding obtained using belief propagation with ordered statistics decoding (BP+OSD) \cite{PhysRevResearch.2.043423}. Belief propagation can also be viewed as obtaining a ``mean-field'' approximation to the optimal decoding problem \cite{midha2026}. Consequently, one would not expect such BP-based decoders to achieve the optimal threshold of QEC codes. However, as BP+OSD remains a practical, fast (time complexity polynomial in $n$) and general decoder for qLDPC codes, it allows us to probe the practical ramifications of such optimally constrained thresholds. In the following results, we use the \emph{LDPC} library to implement the BP+OSD algorithm \cite{Roffe_LDPC_Python_tools_2022}.

In Fig.~\ref{fig:BB_codes_BPOSD}, we simulate the logical failure rate of the weight-6 BB codes with $k=12$ introduced in Ref.~\cite{symons2025} without postselection. We use larger codes in this setting as the smaller codes we investigated in the fully postselected case were significantly affected by finite-size effects. We use $N_{\mathrm{samp}}=10^{5}$ samples for each code and employ the min-sum variant of BP with a maximum number of iterations $n_\mathrm{iter}$ equal to the block sizes of the respective codes $(n)$, with a ``combination sweep'' strategy and order-$0$ postprocessing. The min-sum scaling factor was set at 0.625 \cite{Hillmann2025}. We perform a data collapse by fitting a second-order polynomial to the scaling variable $(p-p_{c})n^{1/\nu}$, where $p_{c}$ is the threshold, $n$ is the code block size (as the precise distances of the codes are not known), and $\nu$ is an effective finite-size rescaling exponent. While the data collapse remains imperfect, the crossing point can be estimated at $p_{c}\approx0.0838(6)$.

\begin{figure}[t]
    \centering
    \includegraphics[width=\linewidth]{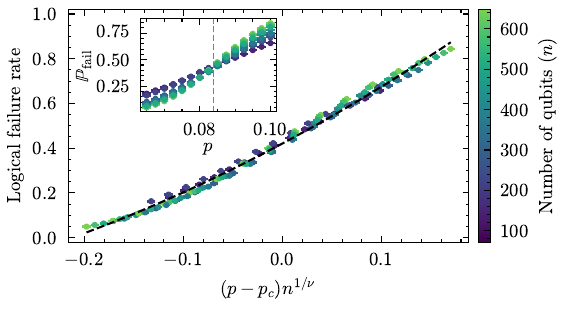}
    \caption{Data collapse of the logical failure rate of the $k=12$ family of weight-6 BB codes from Ref.~\cite{symons2025} with bit-flip noise and no postselection using BP+OSD. \emph{Inset:} the data prior to rescaling.}
    \label{fig:BB_codes_BPOSD}
\end{figure}

In Fig.~\ref{fig:Haahs_code_BPOSD}, we simulate the logical failure rate of Haah's cubic code with periodic boundary conditions on odd lattice sizes with $k=2$ without postselection. Again, we use larger code distances than were studied in the fully postselected setting to mitigate finite-size effects. In these results, we use BP+OSD with OSD-0 postprocessing, with the ``min-sum'' variant of BP and a fixed maximum number of iterations $n_{\mathrm{iter}}=60$. The min-sum scaling factor was set at 0.625. We show a data collapse of the logical failure rate of the odd $L$ under BP+OSD with OSD-0 according to the scaling parameter $x=(p-p_{c})L^{1/\nu}$, where $p_{c}$ is the threshold, and $\nu$ is a finite-size rescaling exponent. The data collapse is again imperfect, but the crossing point can be estimated at $p_{c}\approx 0.0797(4)$.

\begin{figure}[t]
    \centering
    \includegraphics[width=\linewidth]{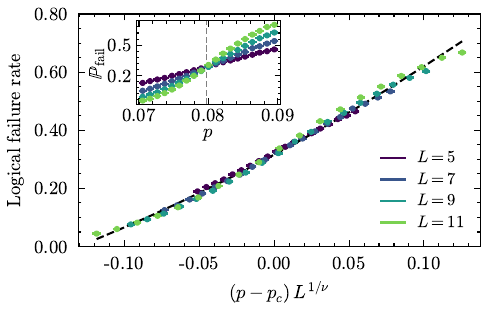}
    \caption{Data collapse of Haah's cubic code with bit-flip noise and no postselection using BP+OSD. \emph{Inset:} the data prior to rescaling.}
    \label{fig:Haahs_code_BPOSD}
\end{figure}

In both cases, the threshold under bit-flip code capacity noise with BP+OSD-0 is around 8\%, which is indeed lower than the optimal prediction of $p_{c}\approx 0.11$ by the arguments above. Moreover, this threshold is also lower than that of the toric code under BP+OSD-0, which was found to possess a numerical threshold of $p_{c}=9.9\%\pm 0.2\%$ \cite{PhysRevResearch.2.043423}. Despite optimal thresholds being largely aligned, the thresholds achieved by practical suboptimal decoders therefore appear to significantly vary. With recent research into belief propagation bounds on tensor network contraction, can we design codes purposefully amenable to BP-based decoding? We leave this question to future work.

\section{Discussion and outlook}\label{sec:discussion}

A recurring theme in fault-tolerant quantum computing is the pursuit of ever-higher thresholds. Our results suggest a complementary perspective: for a broad class of zero-rate \emph{em}-symmetric CSS code families, the optimal code capacity threshold is strongly constrained. While many duality arguments in the statistical mechanical literature are formulated directly in the thermodynamic limit, our derivation required keeping the logical operators and their associated sectors exactly at finite system size before taking the infinite size, zero-rate limit. This makes it possible to see how generalized Kramers-Wannier duality mixes logical sectors, and why this mixing becomes subextensive for zero-rate families. This highlights the value of bringing a QEC perspective to such statistical mechanical problems.

These results shift the natural optimization target from asymptotic threshold values to finite-size performance.
The relevant distinguishing characteristic of a code family becomes the rate at which logical failures are suppressed as the physical error rate is scaled below threshold.
In the present work, we demonstrated that this distinction is governed by a tradeoff between geometry and entropy. The entropy captures the number of low-weight logical operators, and both this entropy and the distance of such logical operators can be determined by the geometry of the code family, which together determine sub-threshold performance.
For example, topological growth produces only polynomially many minimum-weight logical operators, so the entropy of near-minimal failures remains subextensive. By contrast, concatenated growth generically produces an exponential proliferation of minimal logical operators, leading to an entropy that is extensive. As a result, even when two code families share essentially the same threshold, their sub-threshold logical suppression can differ substantially at practical finite sizes.

We emphasize that the above duality constraints are intrinsically tied to zero-rate families. When the asymptotic rate is nonzero, the thermodynamic limit retains an asymptotically large sector structure, and the divergence from self-duality remains extensive (see Corollary~\ref{thm:2}). In this case, the Nishimori point is no longer pinned by the zero-rate condition, and there is no reason to expect different constant rate constructions to align at the same critical physical error rate. Surface codes on self-dual hyperbolic tilings already illustrate this issue, as described in Section~\ref{subsec:phase_diagram_constraints}: despite having constant rate and self-dual microscopic structure, their observed thresholds do not appear to follow the simple finite-rate hashing bound prediction. A more precise theory which incorporates the mixing of logical sectors under generalized Kramers-Wannier duality could provide a more useful statistical mechanical framework for asymptotically good qLDPC codes as well. This leaves open the possibility that good qLDPC families, including constant rate constructions, can outperform two-dimensional topological growth at practical target logical failure rates and overhead.

Future work should extend this duality framework beyond code capacity noise. Phenomenological noise models with noisy syndrome measurements map to higher-dimensional disordered statistical mechanical models, and some of these may possess their own self-dual sectors. For example, the four-dimensional spacetime model associated with the 3D toric code under phenomenological bit-flip noise and measurement errors has its optimal threshold fixed at $p_c \approx 0.11$ by self-duality \cite{Xu_2026}. Treating the mixing of electric and magnetic insertions under this duality more carefully in the phenomenological setting might provide insights for predicting the scaling of logical failure rates. Another avenue for investigation is to extend these results to circuit-level noise. 

It would also be interesting to apply a similar statistical mechanical perspective to the resources required for universal fault tolerance, which has been increasingly investigated in the literature~\cite{xu2026,aitchison2026}. Magic state distillation protocols already have a natural interpretation as recursive maps on noisy resource states, analogous to an RG flow, while transversal gate and non-abelian preparation schemes shift the cost into the structure of the code or state preparation procedure~\cite{PhysRevA.71.022316,PhysRevLett.98.160502,vrty-qs5h}. A statistical mechanical framework comparing these approaches in terms of thresholds, finite-size overheads, and sub-threshold logical suppression could help clarify which approach is most efficient over relevant parameter regimes. 

Finally, although the replica trick and principal Boltzmann factor approximation give remarkably accurate predictions for the phase boundaries across many \emph{em}-symmetric CSS codes, this approximation is not controlled in general. The conditions under which this approximation becomes unreliable should be made precise. Identifying codes for which the principal Boltzmann factor construction does not give a good approximation to the phase boundary could also potentially lead to codes whose thresholds deviate strongly from the hashing bound prediction.

\vspace{2em}
\begin{acknowledgments}
We thank Christopher Chubb, Alexander Cowtan, Timo Hillmann, Grace Sommers and Nicholas O’Dea for helpful discussions.  This work is supported by the ARO through the IARPA ELQ program W911NF-23-2-0223. 
DJW is supported by the Australian Research Council Discovery Early Career Research Award (DE220100625).
\end{acknowledgments}

\bibliography{apssamp}% Produces the bibliography via BibTeX.

\appendix
\widetext

\section{Determination of code phase diagrams}

In this section, we detail the numerical determination of the phase diagram in Fig.~\ref{fig:phase_diagram}. We use the algorithm in Section~\ref{subsec:algorithm} to estimate the partition functions in the logical identity sector $Z_{\vec{e},\vec{0}}(\theta;K,p)$ and the logical $X$ sector $Z_{\vec{e},\vec{0}}(\theta;-K,p)$ for given samples with disorder drawn from i.i.d. $\mathbb{P}(e_{b}=1)$ with probability $p$ and $\mathbb{P}(e_{b}=0)$ with probability $(1-p)$. Above the Nishimori line, we fix values of $p$, and sweep $T$ in a window around an initial $T_{0}$, estimated from the $T_{c}(p)$ of the square-lattice random bond Ising model taken from literature \cite{PhysRevB.65.054425,Wang_03}.

To determine the phase diagram of the models, we do not enforce the Nishimori conditions. Rather, we independently sample $(-1)^{e_{b}}=\pm 1$ with the antiferromagnetic coupling probability $p$, and we fix $T\in[0,\infty)$. The algorithm in Section~\ref{subsec:algorithm} is straightforwardly adapted to separate $p$ and $T$, i.e., off the Nishimori line. We define a sample-dependent ``failure probability'' whose quenched average along the Nishimori line becomes the optimal logical failure rate, see Eq.~(5) of Ref.~\cite{Chen_2025}
\begin{equation}
    \mathbb{P}_{\mathrm{fail}}(\vec{e},p,T)=1-\frac{\max(Z_{\vec{e},\vec{0}}(\theta;K,p),Z_{\vec{e},\vec{0}}(\theta;-K,p))}{Z_{\vec{e},\vec{0}}(\theta;K,p)+Z_{\vec{e},\vec{0}}(\theta;-K,p)}.
\end{equation}
The disorder-averaged order parameter is then
\begin{equation}
\mathbb{P}_{\mathrm{fail}}(p,T)
=\langle\mathbb{P}_{\mathrm{fail}}(\vec{e},p,T) \rangle_{\vec{e}}.
\end{equation}

Numerically, for each fixed $p$ we evaluate $\mathbb{P}_{\mathrm{fail}}(p,T)$ over a temperature grid and concatenation levels $r\in\{2,3,4,5\}$ (equivalently $d=3^r$), with typically $N_{\mathrm{samp}}=10^5$ samples per $(p,T)$ point in the temperature-sweep runs. We then extract the critical temperature $T_c(p)$ via finite-size scaling collapse with the ansatz \cite{Wang_03}
\begin{equation}
\mathbb{P}_{\mathrm{fail}}(p,T)=
f\!\left((T-T_c)\,d^{1/\nu}\right),
\end{equation}
and approximate $f$ by a quadratic polynomial,
\begin{equation}
f(x)\approx A x^2 + B x + C.
\end{equation}
The parameters $(T_c,\nu,A,B,C)$ are obtained by minimizing the weighted least squares objective
\begin{equation}
\chi^2=
\sum_i
\frac{\left[y_i-(A x_i^2+B x_i+C)\right]^2}{\sigma_i^2},
\qquad
x_i=(T_i-T_c)\,d_i^{1/\nu},
\end{equation}
where $y_i=\mathbb{P}_{\mathrm{fail},i}$ and $\sigma_i$ is the Monte Carlo standard error. Repeating this procedure for multiple fixed $p$ yields the phase boundary $T_c(p)$ in the $(p,T)$ plane. 

We plot in Figs.~\ref{fig:threshold_RSC} and \ref{fig:threshold_Steane} data collapse plots of $\mathbb{P}_{\mathrm{fail}}(p,T)$ along the Nishimori line (i.e., under optimal MLD) for the concatenated surface-17 code and concatenated Steane code, respectively. We use the algorithm detailed in Section~\ref{subsec:algorithm} to obtain the logical failure rate.

\begin{figure*}[t]
    \centering
    \subfloat[\label{fig:threshold_RSC}]
    {\includegraphics[width=0.5\linewidth]{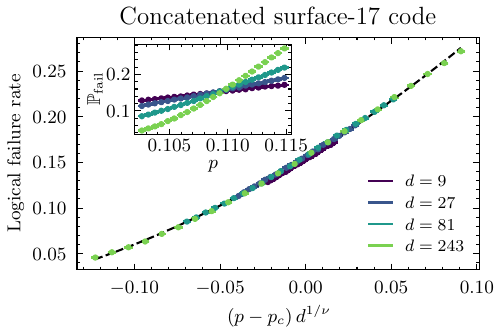}}
    \hfill
    \subfloat[\label{fig:threshold_Steane}]
    {\includegraphics[width=0.5\textwidth]{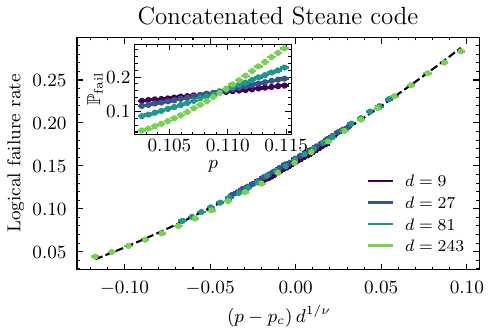}}

    \caption{
    Logical failure rate $\mathbb{P}_{\mathrm{fail}}$ for the concatenated codes under bit-flip noise in the code capacity setting, plotted as a function of the scaling variable $x=(p-p_{c})d^{1/\nu}$ to demonstrate data collapse, where $d$ is the code distance.
    (a) Concatenated surface-17 code.
    (b) Concatenated Steane code.
    In both panels, the dashed black curve shows the best fit to a second order polynomial in $x$, treating the three polynomial parameters, the critical error rate $p_{c}$, and the critical exponent $\nu$ as free fit parameters.
    The scatter points represent numerical simulation results with $N_{\mathrm{samp}}=10^{5}$ samples.}
    \label{fig:curvefit_combined_concat}
\end{figure*}

For completeness, we also plot in Figs. \ref{fig:threshold_rotated_surface} and \ref{fig:threshold_color} data collapse plots of $\mathbb{P}_{\mathrm{fail}}(p,d)$ along the Nishimori line (i.e., under optimal MLD) for the topological rotated surface code and (6.6.6) color code, respectively. The (6.6.6) color code failure rate (which is only approximately optimal because of finite bond truncation) is computed using the Python package \emph{qecsim} \cite{qecsim}, with an MPS decoder with bond dimension $\chi=10$. The rotated surface code failure rate is computed using the Python package \emph{planar} (optimal) \cite{PhysRevLett.134.190603}.

\begin{figure*}[t]
    \centering
    \subfloat[\label{fig:threshold_rotated_surface}]
    {\includegraphics[width=0.5\textwidth]{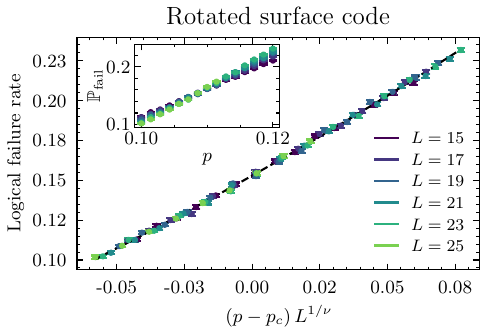}}
    \hfill
    \subfloat[\label{fig:threshold_color}]
    {\includegraphics[width=0.5\linewidth]{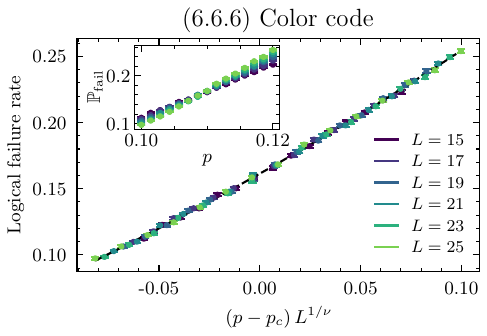}}

    \caption{
    Logical failure rate $\mathbb{P}_{\mathrm{fail}}$ for the topological codes under bit-flip noise in the code capacity setting, plotted as a function of the scaling variable $x=(p-p_{c})L^{1/\nu}$ to demonstrate data collapse, where $L$ is the code distance.
    (a) Rotated surface code.
    (b) $(6.6.6)$ color code.
    In both panels, the dashed black curve shows the best fit to a second order polynomial in $x$, treating the three polynomial parameters, the critical error rate $p_{c}$, and the critical exponent $\nu$ as free fit parameters.
    The scatter points represent numerical simulation results with $N_{\mathrm{samp}}=10^{5}$ samples.}
    \label{fig:curvefit_combined_topo}
\end{figure*}

\section{Concatenated codes and real-space renormalization group flow}\label{appendix:RG_flow}

\begin{theorem}\label{eq:theorem1}
Let $\mathcal{C}$ be a fully postselected $[[n,1,d]]$ stabilizer code in the code capacity bit-flip setting. Suppose its logical $\overline{X}$ has a transversal representative so that the two logical sectors satisfy $Z_{\vec{1},\vec{0}}(\theta;K_{p})=Z_{\vec{0},\vec{0}}(\theta;-K_{p})$, and concatenation acts by independent hierarchical substitution of the base interaction. Then concatenation defines an exact real-space renormalization group recursion for the clean statistical mechanical model.

If $f(p)=\mathbb{P}_{\mathrm{fail}}^{(1)}(p)$ denotes the base code logical failure map, then after $r$ levels of concatenation,
\[
\mathbb{P}_{\mathrm{fail}}^{(r)}(p)=f^{\circ r}(p).
\]
\end{theorem}

Here we outline a proof of this theorem, showing one level of concatenation renormalizes the coupling constant with a constant prefactor. The same argument applies to establish that the result holds inductively for higher levels of concatenation. 

\begin{proof}
For an $[[n,1,d]]$ seed code with transversal logical Pauli representatives, the interaction matrix (parity-check matrix) $\theta$ transforms under a single step of concatenation through
\begin{equation}\label{eq:theta_concat}
    \theta^{(r)}=\begin{bmatrix}
        \theta^{(r-1)}\otimes \vec{1}_{N_{b}}^{\top} \\
        \mathds{1}_{N_{b}^{(r-1)}}\otimes \theta
    \end{bmatrix}
    \end{equation}
where $\vec{1}_{N_{b}}$ is a $N_{b}\times 1$ column vector of ones, and $\mathds{1}_{N_{b}^{(r-1)}}$ is the $N_{b}^{(r-1)}\times N_{b}^{(r-1)}$ identity. As each interaction at the base level is replaced with a copy of the original Hamiltonian, we obtain the form above. Under concatenation we observe $N^{(r)}_{s}= N_{s}^{(r-1)}+N_{b}^{(r-1)}N_{s}$ and $N_{b}^{(r)}= (N_{b}^{(r-1)})N_{b}=(N_{b})^{r}$. In other words, in each concatenation level, each physical qubit is replaced with a logical qubit comprised of $N_{b}$ physical qubits. We retain all of the original stabilizers, but add $N_{s}$ new stabilizers onto each of the original physical qubits.

For a single layer of concatenation, we have for the logical identity coset
\begin{equation}\label{eq:concatenated_ZI}
    Z^{(2)}_{\vec{0},\vec{0}}(\theta;K_{p})=\sum_{S_{i}^{(2)}=\pm 1}\sum_{S_{i}^{(1)}=\pm 1}\prod_{b^{(2)}=1}^{n}\exp\left(\prod_{j}S_{j}^{\theta_{jb}}K_{p}\left[\sum_{b^{(1,b^{(2)})}=1}\prod_{j^{(2,b^{(2)})}}^{n}S_{j^{(2,b^{(2)})}}^{\theta_{j^{(2,b^{(2)})}b^{(2,b^{(2)})}}}\right]\right),
\end{equation}
where the $S_{i}^{(2)}$ are the outer spins, and the $S_{i}^{(1)}$ are the inner spins. When the outer spin is fixed in some configuration $\vec{S_{i}}^{(2)}$ such that $\prod_{j}S_{j}^{\theta_{jb}}=+1$, the local Boltzmann factor (after summing over the inner spins $S_{i}^{(1)}$) can be seen to be equal to $Z_{\vec{0},\vec{0}}(\theta;K_{p})$, and when $\prod_{j}S_{j}^{\theta_{jb}}=-1$, the Boltzmann factor equals $Z_{\vec{0},\vec{0}}(\theta;-K_{p})$. In other words, when the outer spins multiply to give $+1$, the local Boltzmann factor is equal to the inner model's identity partition function, and when they multiply to give $-1$, we obtain the inner model's logical $X$ partition function.

Next, let us write down this local Boltzmann factor
\begin{equation}
    x(\prod_{j}S_{j}^{\theta_{jb}})=Ae^{K_{p}'\prod_{j}S_{j}^{\theta_{jb}}}
\end{equation}
where $A$ is some factor and $K'_{p}$ is the renormalized dimensionless coupling constant. We find $A$ by multiplying the two Boltzmann factors
\begin{equation}
    x(+1)x(-1)=A^{2}=Z_{\vec{0},\vec{0}}(\theta;K_{p})Z_{\vec{0},\vec{0}}(\theta;-K_{p})\implies A=\sqrt{Z_{\vec{0},\vec{0}}(\theta;K_{p})Z_{\vec{0},\vec{0}}(\theta;-K_{p})}.
\end{equation}
And by dividing them we obtain 
\begin{align}
    x(+1)/x(-1)=e^{2K'_{p}}=\frac{Z_{\vec{0},\vec{0}}(\theta;K_{p})}{Z_{\vec{0},\vec{0}}(\theta;-K_{p})}\implies K_{p}^{'}=\frac{1}{2}\log\left(\frac{Z_{\vec{0},\vec{0}}(\theta;K_{p})}{Z_{\vec{0},\vec{0}}(\theta;-K_{p})}\right).
\end{align}
Noting that $\mathbb{P}_{\mathrm{fail}}=\frac{Z_{\vec{0},\vec{0}}(\theta;-K_{p})}{Z_{\vec{0},\vec{0}}(\theta;+K_{p})+Z_{\vec{0},\vec{0}}(\theta;-K_{p})}$, we have $\frac{1-\mathbb{P}_{\mathrm{fail}}}{\mathbb{P}_{\mathrm{fail}}}=\frac{Z_{\vec{0},\vec{0}}(\theta;K_{p})}{Z_{\vec{0},\vec{0}}(\theta;-K_{p})}$. In other words, the Nishimori conditions have the same functional form but instead of the i.i.d. bit-flip probability $p$, we have the logical failure rate of the inner code $\mathbb{P}_{\mathrm{fail}}(p)$. Now accounting for all $n$ copies of the code on the inner layer, we can rewrite Eq.~\eqref{eq:concatenated_ZI} through
\begin{align}
    Z^{(2)}_{\vec{0},\vec{0}}(\theta;K_{p})&=\sum_{S_{i}^{(2)}=\pm 1}\prod_{b^{(2)}=1}^{n}\sqrt{Z_{\vec{0},\vec{0}}(\theta;+K_{p})Z_{\vec{0},\vec{0}}(\theta;-K_{p})}\exp(K'_{p}\prod_{j}S_{j}^{\theta_{jb}})\\
    &=\left(\sqrt{Z_{\vec{0},\vec{0}}(\theta;+K_{p})Z_{\vec{0},\vec{0}}(\theta;-K_{p})}\right)^{n}\sum_{S_{i}^{(2)}=\pm 1}\prod_{b^{(2)}=1}^{n}\exp(K'_{p}\prod_{j}S_{j}^{\theta_{jb}})
\end{align}
where the same argument gives a similar formula for the logical $X$ coset, i.e., $Z_{\vec{0},\vec{0}}^{(2)}(\theta;-K_{p})$. In taking the failure rate through the ratio of partition functions, the prefactor cancels and we are left with an exact RG flow with $\mathbb{P}_{\mathrm{fail}}^{(2)}(K_{p})=\mathbb{P}_{\mathrm{fail}}^{(1)}(K'_{p})$. In other words, $\mathbb{P}_{\mathrm{fail}}^{(2)}(p)=\mathbb{P}_{\mathrm{fail}}(\mathbb{P}_{\mathrm{fail}}(p))$.

\end{proof}

The intuition we take from this analysis is that code concatenation renormalizes the effective coupling $K_{p}$, while the functional form of the failure probability remains that of the base code, and the prefactors generated by integrating out inner degrees of freedom cancel exactly in the logical failure ratio. A corollary of this is that we can obtain the threshold through the unstable fixed point of the RG flow of either the coupling constant $K_{p}$ or the flow of $\mathbb{P}_{\mathrm{fail}}$. This RG flow of $\mathbb{P}_{\mathrm{fail}}$ can schematically be depicted as in Fig.~\ref{fig:RGflow}.

\begin{figure}
    \centering
            \begin{tikzpicture}[
  thick,
  >=Latex,
  dot/.style={circle, inner sep=0pt, minimum size=4pt},
]
  % axis
  \draw (0,0) -- (6,0);

  % fixed points
  \node[dot, fill=black, label=below:{$0$}]            (p0)  at (0,0) {};
  \node[dot, draw=black, fill=white, label=below:{$p_c$}] (pc)  at (3,0) {};
  \node[dot, fill=black, label=below:{$\tfrac12$}]     (ph)  at (6,0) {};

  % RG flow arrows (away from p_c)
  \draw[->] (pc) ++(-0.2,0.30) - ++(-2.6,0);
  \draw[->] (pc) ++( 0.2,0.30) - ++( 2.6,0);
\end{tikzpicture}
    \caption{Renormalization group flow of $\mathbb{P}_{\mathrm{fail}}$, possessing three fixed points: $\mathbb{P}_{\mathrm{fail}}=0$ (stable; below threshold), $\mathbb{P}_{\mathrm{fail}}=1/2$ (stable; above threshold), $\mathbb{P}_{\mathrm{fail}}(p)=p=:p_{c}$ (unstable; threshold).}
    \label{fig:RGflow}
\end{figure}
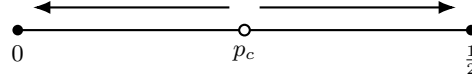

\section{Population annealing simulations for A2BGA codes}\label{app:MCMC_pop_annealing}

In this appendix we detail the numerical simulation procedure used to obtain the logical failure rates in the fully postselected limit for the A2BGA codes investigated in Section~\ref{sec:BB_codes}.  The calculation is performed in the clean statistical mechanical model obtained after full postselection.  For each code, we estimate the partition functions in the different logical sectors using population annealing Markov chain Monte Carlo (MCMC), and we direct the reader to Ref.~\cite{PhysRevE.92.063307} for a pedagogical introduction to the technique.

Population annealing was run independently in each logical sector.  We initialized a population of $N_{R}$ replicas at $K=0$, where the distribution is uniform, and annealed the population through a nonuniform grid of $N_{K}$ values of $K$ between $K_{\min}=0$ and $K_{\max}=0.9$. At each step $K_{\tau}\to K_{\tau+1}$, replicas were reweighted by
\begin{equation}
    w_{j}=\exp\left[-(K_{\tau+1}-K_{\tau}) E_{j}\right],
\end{equation}
where $E_{j}$ is the energy of replica $j$ according to the Hamiltonian expressed through Eq.~\eqref{eq:GIM}.  The population was then resampled using systematic resampling, followed by a fixed number of Metropolis sweeps at the new coupling. The free energy increments accumulated during the reweighting steps give an estimate of $\log Z_{\lambda}(\theta;K)$ for each sector.

After estimating all logical sector partition functions, we compute the sector probabilities
\begin{equation}
    \mathbb{P}_{\lambda}(\theta;K)=\frac{Z_{\lambda}(\theta;K)}{\sum_{\lambda'} Z_{\lambda'}(\theta;K)} .
\end{equation}
The fully postselected block failure probability is then
\begin{equation}
    \mathbb{P}_{\mathrm{fail}}(\theta;K) = 1 - \mathbb{P}_{\vec{0}}(\theta;K).
\end{equation}

The numerical parameters used for the population annealing simulations are listed in Tab.~\ref{tab:pa_simulation_parameters}.  Here $N_{R}$ is the population size, $n_{\mathrm{sw}}$ is the number of Metropolis sweeps after each resampling step, $N_{K}$ is the number of annealing temperatures, and $N_{\mathrm{PA}}$ is the number of independent population annealing replicates.  The simulations were run independently for all $2^{k}$ logical sectors, except for the $L=2$ Haah code point,
which was evaluated by exact enumeration.

\begin{table*}[t]
    \begin{ruledtabular}
    \begin{tabular}{c|c|c|c|c|c|c}
        Code & $k$ & Sectors & $N_{\mathrm{PA}}$ & $N_{R}$ & $n_{\mathrm{sw}}$ & $N_K$ \\
        \hline
        \multicolumn{7}{c}{Haah's cubic code} \\
        \hline
        $L=2$ & 6 & 64 & -- & -- & -- & -- \\
        $L=3$ & $2$ & $4$ & $1$ & $32768$ & $20$ & $800$ \\
        $L=5$ & $2$ & $4$ & $1$ & $262144$ & $50$ & $1500$ \\
        \hline
        \multicolumn{7}{c}{Bivariate bicycle codes} \\
        \hline
        $[[18,4,4]]$ & $4$ & $16$ & $1$ & $32768$ & $20$ & $800$ \\
        $[[28,6,4]]$ & $6$ & $64$ & $1$ & $65536$ & $30$ & $1000$ \\
        $[[30,4,6]]$ & $4$ & $16$ & $1$ & $65536$ & $30$ & $1000$ \\
        $[[36,4,6]]$ & $4$ & $16$ & $1$ & $65536$ & $40$ & $1200$ \\
        $[[54,4,8]]$ & $4$ & $16$ & $1$ & $131072$ & $60$ & $1600$ \\
    \end{tabular}
    \end{ruledtabular}
    \caption{
    Population annealing parameters used to estimate fully postselected logical failure rates.  Simulations were performed on a piecewise uniform grid in $K\in[0,0.9]$, with increased resolution near the transition region.  Each logical sector was simulated independently.
    }
    \label{tab:pa_simulation_parameters}
\end{table*}

\end{document}